\documentclass[twocolumn,10pt]{IEEEtran}
\usepackage{braket}

\usepackage{hyperref}
\hypersetup{colorlinks,urlcolor=black}
\usepackage{booktabs}

\input{def.tex}
\usepackage{dsfont}
\usepackage{minibox}
\DeclareSymbolFont{matha}{OML}{txmi}{m}{it}% txfonts
\DeclareMathSymbol{\varv}{\mathord}{matha}{118}
\usepackage{eurosym}
\usepackage{hhline,physics}

\usepackage{multicol}
\usepackage{algorithm}
\usepackage{graphicx}
\usepackage{textcomp}
\usepackage{pifont}
\usepackage{float}

\usepackage{subcaption}
\usepackage[multiple]{footmisc}

\usepackage{algorithmic}

\makeatletter
\renewcommand{\baselinestretch}{0.98}

\IEEEoverridecommandlockouts

\begin{document}
	\title{MISO Downlink with Fluid Antenna Multiple Access} 
	\author{
	Anastasios~Papazafeiropoulos,~\IEEEmembership{Senior Member,~IEEE} 		
	%	\thanks{M. Diamanti, A. Papazafeiropoulos, and S. Papavassiliou are with the Institute of Communication and Computer Systems (ICCS), School of Electrical and Computer Engineering, National Technical University of Athens, Zografou, Greece, 15780, e-mail: mdiamanti@netmode.ntua.gr; apapazafeiropoulos@mail.ntua.gr; papavass@mail.ntua.gr.}
	
	\thanks{Anastasios Papazafeiropoulos is with the Communications and Intelligent Systems Research Group, University of Hertfordshire, Hatfield AL10 9AB, U. K. (e-mail:tapapazaf@gmail.com). \par 
		%			Pandelis~Kourtessis is with the Communications and Intelligent Systems Research Group, University of Hertfordshire, Hatfield AL10 9AB, U. K. (e-mail: p.kourtessis@herts.ac.uk). \par 
	}
}
	\maketitle\vspace{-1.9cm}
\begin{abstract}
Fluid antenna multiple access (FAMA) enables each user to rapidly switch
among several closely spaced ports and select the strongest received signal.
Although this mechanism offers micro-scale spatial diversity, its behavior in
multiuser downlink systems with spatial correlation and linear precoding is
not well understood. This paper develops a unified analytical framework for
the multiple-input single-output (MISO) downlink with FAMA users served via maximum ratio transmission
(MRT) or zero-forcing (ZF). \textcolor{black}{We show that the per-port signal-to-interference ratio
	(SIR) follows a Beta-prime distribution with parameters
	\((M_{\mathrm{eff}},L)\), where \(M_{\mathrm{eff}}=M\) under MRT and
	\(M_{\mathrm{eff}}=M-U+1\) under ZF}, and derive closed-form finite-sum cumulative distribution functions (CDFs) for both
cases. We further provide the first analytical characterization of cross-port
SIR correlation. \textcolor{black}{Furthermore, we derive rigorous outage probability bounds that tightly bracket the exact performance and become exact in the limiting cases of fully correlated and independent ports.} Asymptotic analyses reveal the
fundamental diversity orders and tail behavior for each precoder. Numerical
results confirm the accuracy of the SIR distributions, correlation model, and
outage bounds, and show that MRT achieves weaker port correlation and larger
selection gains than ZF when the base station (BS) has ample spatial degrees of freedom.
The framework offers explicit guidelines for port configuration and precoder
selection in practical FAMA systems.
\end{abstract}
	
	
\section{Introduction}
Multiple-input multiple-output (MIMO) technology has fundamentally transformed wireless communication by exploiting spatial dimensions to enhance reliability and throughput without additional bandwidth or time resources. By leveraging spatially diverse channel paths, MIMO achieves significant performance gains in both point-to-point~\cite{Telatar1999,Goldsmith2005} and multiuser systems, where spatial multiplexing enables simultaneous service of multiple users within the same time–frequency block~\cite{Caire2003,Jindal2006}. \textcolor{black}{While linear precoding in multiuser MIMO is well understood, its interaction with spatially reconfigurable user-side antennas remains largely unexplored. In particular, conventional MIMO models do not capture how antenna mobility and port selection affect  signal-to-interference 	ratio (SIR) statistics and achievable diversity in multiuser downlink systems.}

Fluid antennas are electronically controlled radiating structures whose geometry or position can be reconfigured to modify electromagnetic characteristics without deploying multiple fixed antenna elements. Although formally introduced in~\cite{Kar2010}, related concepts date back earlier works such as~\cite{Kosta2004}. \textcolor{black}{Fluid antennas have been realized through various liquid-metal and reconfigurable radiating structures, as surveyed in~\cite{Paracha2019,Huang2021}.} \textcolor{black}{
	Also, fluid antenna concepts have also been investigated in the context of index modulation and robust transmission schemes. For example, fluid antenna–assisted index modulation has been studied for MIMO systems and reconfigurable intelligent surface (RIS)-aided  millimeter wave (mmWave) scenarios, with emphasis on modulation design, detection complexity, and robustness to hardware impairments~\cite{Zhu2024,Zhu2024a,Guo2025}. These works differ fundamentally from the present study, as they do not consider multiuser downlink precoding, per-port SIR and outage characterization, or the impact of spatial correlation and port selection under interference-limited operation.
}

More \textcolor{black}{recently}, fluid antenna multiple access (FAMA) has emerged as a promising paradigm in which users dynamically reposition their antenna within a compact aperture and select the port that maximizes instantaneous performance~\cite{FAMA_MultipleAccess1}. Unlike conventional MIMO, which relies on simultaneous multi-antenna observations, FAMA exploits micro-scale spatial variations through single-RF-chain port switching. \textcolor{black}{This distinctive operating principle introduces new analytical challenges, since the per-port received signals are inherently correlated and the resulting performance depends critically on both the antenna displacement model and the adopted multiuser transmission strategy at the  base station (BS).} Existing analytical studies mainly focus on single-input single-output (SISO) configurations, independent ports, or strongest-port statistics, without accounting for multiuser precoding or spatial correlation.

Recent works have investigated slow and fast FAMA protocols~\cite{Wong2023a,Wong2023b} and OFDM-based integration into 5G-NR systems~\cite{Hong2025}. \textcolor{black}{In parallel, opportunistic user scheduling has been combined with FAMA to relax the requirement on the number of ports per fluid antenna while preserving high multiplexing gains, by selecting a small subset of users with favorable FAMA-induced SIR from a larger pool \cite{Wong2023c}. FAMA has also been embedded into integrated data and energy transfer (IDET) architectures, where port selection is optimized either for SINR or harvested power and closed-form outage and multiplexing expressions are derived for joint wireless data transfer (WDT) and wireless energy transfer (WET) operation \cite{Lin2025}. Also,  it has  been combined with coded caching to mitigate interference and worst-user bottlenecks, enabling interference-free multiuser connectivity with a single transmit antenna and fluid-antenna port selection, with closed-form rate and scaling characterizations \cite{Zhao2026}. At a more fundamental level, slow-FAMA has been studied for unsourced random access (URA), where MIMO fluid-antenna systems achieve notable capacity and error-floor improvements over classical MIMO-URA bounds \cite{Zhang2025}.  Moreover, fluid antenna systems have been shown to enhance both orthogonal and non-orthogonal multiple access schemes by exploiting fine-grained spatial sampling to mitigate multiuser interference and improve fairness compared with fixed array baselines \cite{New2024}.} 
\textcolor{black}{Despite these advances, the existing literature has largely emphasized protocol design, system integration, or single-link analysis, leaving the fundamental performance limits of multiuser downlink FAMA under linear precoding analytically unexplored.}  The joint effects of  maximum ratio transmission
(MRT)/zero-forcing (ZF) beamforming, interference, port-wise SIR statistics, spatial correlation, and outage performance remain analytically uncharacterized in closed form.

Motivated by these gaps, this paper develops a unified analytical framework for downlink MISO FAMA under MRT and ZF precoding.\footnote{\textcolor{black}{
		From a practical perspective, fluid antenna systems enable spatial diversity and interference mitigation using a single radio frequency (RF) chain, avoiding the cost, size, and power overhead of conventional multi-antenna arrays. By exploiting micro-scale spatial variations within a compact aperture, FAMA is particularly attractive for small-form-factor devices and dense multiuser scenarios where traditional MIMO is constrained.
	}\textcolor{black}{
		Although fluid-antenna technology is still evolving, recent experiments and system-level studies demonstrate the feasibility of rapid port switching and practical integration, motivating analytical frameworks that characterize the fundamental performance limits of FAMA under realistic channel conditions.
}}
We derive exact per-port SIR distributions, characterize cross-port correlation, and establish outage bounds for arbitrary dependence. Asymptotic analysis reveals how FAMA gains scale with the number of ports, the effective BS array dimension, and the interference level. \textcolor{black}{Beyond its analytical novelty, the proposed framework provides closed-form performance benchmarks for multiuser downlink FAMA systems with linear precoding, enabling rapid and reproducible evaluation without reliance on Monte-Carlo (MC) simulations.} \textcolor{black}{The analytical results offer direct design-level insight into the relative merits of MRT and ZF precoding, the impact of port correlation on selection diversity, and the role of the fluid-antenna aperture in determining achievable FAMA gains.}\footnote{\textcolor{black}{
Practical fluid antenna implementations may be affected by mechanical imperfections and actuation inaccuracies. The present work does not model these hardware aspects, but instead characterizes the fundamental channel-level performance limits of FAMA under spatial correlation and multiuser precoding, providing a theoretical benchmark that isolates signal-processing effects. 
The adopted spatial correlation model serves as an abstraction of port-position uncertainty. Additional mechanical non-idealities would increase correlation and reduce selection gains, so the reported results should be interpreted as optimistic upper bounds on performance.
}}

\textcolor{black}{
	Although MRT/ZF precoding and Beta-prime SIR distributions under Rayleigh fading are individually well understood, their combination with fluid-antenna port selection in a multiuser downlink setting introduces fundamentally new analytical challenges. Precoding is designed from a single reference port, while performance depends on opportunistic selection over correlated ports, inducing coupling between precoding, interference, and correlation that is not captured by existing diversity or selection analyses.
}

\paragraph*{Comparison to Existing FAMA Studies}

Representative works such as~\cite{FAMA_MultipleAccess1} and OFDM-based FAMA studies in~\cite{Hong2025} focus on protocol design, system integration, or link-level evaluation, without providing closed-form per-port SIR characterization under multiuser precoding and correlated port selection. In contrast, this work develops a unified framework integrating multiuser MISO precoding, exact per-port SIR distributions, spatial correlation modeling, and closed-form outage characterization.\footnote{\textcolor{black}{While this work focuses on SIR and outage-based performance metrics, the derived SIR distributions also provide a direct analytical foundation for evaluating ergodic rate, sum throughput, and fairness metrics in FAMA-enabled multiuser systems.}}

\textcolor{black}{
	Unlike classical correlated diversity or antenna-selection analyses, which assume simultaneous observation of multiple branches or independent combining, the FAMA setting considered here involves port-wise probing with a single RF chain and precoding mismatch across ports. As a result, existing results on correlated selection diversity cannot be directly adapted to obtain the per-port SIR statistics, cross-port correlation, or outage behavior derived in this work.
}

\textcolor{black}{
	Existing studies on FAMA have primarily focused on SISO configurations, typically assuming independently fading fluid antenna ports and without considering multiuser downlink precoding. Consequently, the interaction between spatial correlation across fluid antenna ports and linear precoding strategies has not been addressed in prior work. To the best of the authors' knowledge, no existing study has analyzed per-port SIR statistics of FAMA systems under MRT or ZF precoding, while jointly accounting for spatial correlation and its impact on selection-based outage probability. The present work fills this gap by providing the first analytical treatment of these effects in multiuser downlink FAMA systems.
}

\subsection*{Main Contributions}

The main contributions of this work are summarized as follows:

\begin{itemize}
	
	\item \textbf{Exact per-port SIR distributions under MRT and ZF:}
	We show that, under both precoders, the per-port SIR follows a
	Beta-prime distribution with parameters $(a,b)$ determined by the effective
	BS array dimension. \textcolor{black}{Specifically, \((a,b)=(M_{\mathrm{eff}},L)\), where
		\(M_{\mathrm{eff}}=M\) for MRT and \(M_{\mathrm{eff}}=M-U+1\) for ZF,
		leading to markedly different tail and low-threshold behaviors.}
	
	\item \textbf{Closed-form finite-sum cumulative distribution function (CDF) expressions:}
	The incomplete Beta representations are converted into finite-sum CDFs
	that provide exact tractable formulas for the SIR distribution at any
	port, enabling rigorous outage analysis without numerical integration.
	
	\item \textbf{Analytical cross-port SIR correlation model:}
	We derive the first closed-form expression for the correlation between
	SIRs across fluid-antenna ports. The correlation coefficient depends
	jointly on the geometric overlap parameters $\mu_k\mu_\ell$ and the
	effective BS array dimension $M_{\mathrm{eff}}$, providing an explicit
	link between aperture size, precoding choice, and selection diversity.
	
	\item \textbf{Rigorous FAMA outage bounds for arbitrary correlation:}
	Using the per-port CDF, we derive tight upper and lower outage bounds
	that remain valid without any independence assumptions. These bounds
	clarify how port correlation limits the achievable selection gain,
	and quantify the regimes where the FAMA gain approaches the ideal
	independent-port case.
	
	\item \textbf{Asymptotic SIR and outage analysis:}
	We obtain small-threshold expansions, large-SIR tail behavior, and
	large-$(N,M)$ scaling laws that provide fundamental insight into the
	dependence of the FAMA gain on the number of ports, aperture width,
	antenna count, and interference level. \textcolor{black}{The asymptotics reveal a unified single-port diversity order
		\(M_{\mathrm{eff}}\), with \(M_{\mathrm{eff}}=M\) for MRT and
		\(M_{\mathrm{eff}}=M-U+1\) for ZF, and a universal
		interference-controlled tail exponent \(L\).}
	
	\item \textbf{Design guidelines for MISO FAMA systems:}
	Based on our analytical results, we provide engineering insights on
	when MRT or ZF is preferable, how large the fluid-antenna aperture
	should be to ensure weak correlation, and how many ports offer
	meaningful gains under given geometric and system constraints.
\end{itemize}

	The remainder of the paper is organized as follows. Section~\ref{system} introduces the
	system model, including the channel representation,
	the MRT and ZF precoders, and the channel state information (CSI) assumptions. Section~\ref{statistical} derives the
	exact per-port SIR distributions and the associated Beta-prime expressions for
	both precoders. Section~\ref{fama} develops the cross-port SIR correlation model,
	followed by rigorous outage bounds and asymptotic analyses characterizing the
	diversity order and tail behavior. Section~\ref{Numerical} presents numerical results that
	validate all analytical findings and provide practical design insights for
	MISO FAMA systems. Finally, Section~\ref{conclusion} concludes the paper and outlines
	directions for future research.
	
	\paragraph*{Notations}
	Bold uppercase letters denote matrices, whereas bold lowercase letters represent vectors. 
	For any matrix $\mathbf{X}$, the operators $\mathbf{X}^{\mathsf{T}}$ and $\mathbf{X}^{\mathsf{H}}$  
	 correspond to the transpose and Hermitian transpose, respectively. The expectation and variance operators are 
	written as $\mathbb{E}\{\cdot\}$ and $\mathrm{Var}\{\cdot\}$. Also,  $J_0(\cdot)$ is the zero-order Bessel function of the first kind,
	and $\mathbf{I}_M$ is the $M\times M$ identity matrix.

	\section{System Model}\label{system}
	We consider a downlink multiuser MISO system, where a BS with
	$M$ antennas serves $U$ single-RF users, each equipped with a fluid antenna
	capable of switching among $N$ ports within a small aperture, as shown in Fig. \ref{Fig0}. Each user
	selects the port that maximizes its instantaneous SIR, while the BS employs
	linear precoding based on the CSI of one reference port per user. The
	resulting performance depends critically on the spatial correlation across
	fluid-antenna ports, the displaced complex Gaussian channel model, and the
	choice of MRT or ZF beamforming. This section introduces the geometric
	correlation model, downlink signal model, and CSI assumptions required for
	the analysis that follows.
		\begin{figure}[!h]
		\begin{center}
			\includegraphics[width=0.7\linewidth]{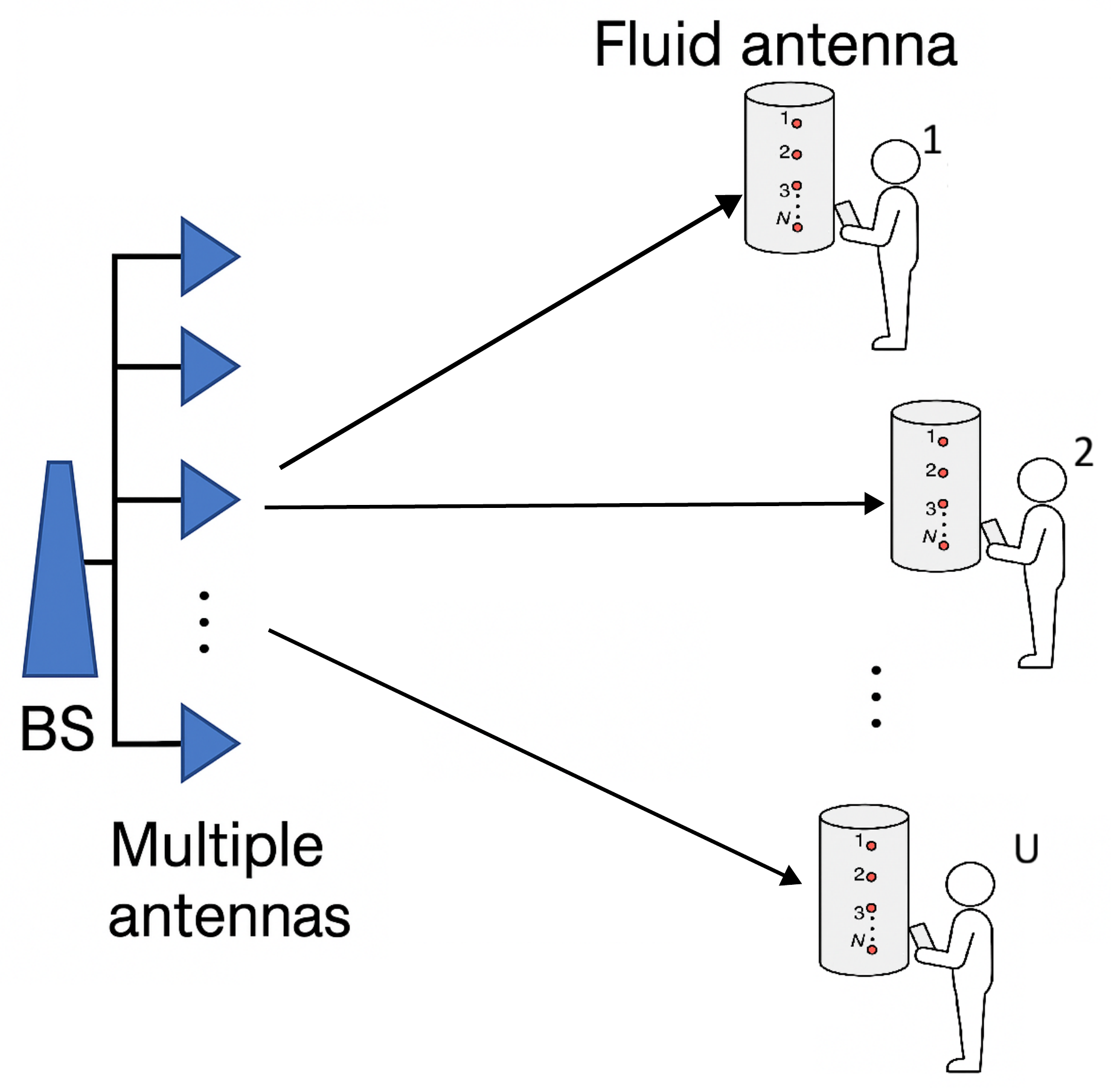}
			\caption{\footnotesize{ A MISO downlink system model with FAMA. }}
			\label{Fig0}
		\end{center}
	\end{figure}

	\subsection{Fluid Antenna Geometry and Spatial Correlation}\label{correlation}
	
	The $N$ ports of the fluid antenna at user $u$ are positioned along a
	one-dimensional aperture of length $W\lambda$, where $\lambda$ denotes the
	carrier wavelength, and $W$ is a  parameter that specifies the aperture length in wavelengths. The displacement of the $k$-th port from a reference
	location is modeled as
	\begin{equation}
		d_k = \frac{k-1}{N-1} W\lambda, \quad k = 1,2,\ldots,N,
	\end{equation}
	so that the ports are uniformly distributed over the aperture.
	
	We assume an isotropic scattering environment in the azimuth plane and
	model the spatial correlation of the complex baseband channel at different
	ports via the well-known Bessel-type correlation function. Let
	$\mathbf{h}_{u,k} \in \mathbb{C}^{M\times 1}$ denote the MISO channel vector
	from the BS antenna array to the $k$-th port of user $u$. The correlation
	between channel entries at ports $k$ and $\ell$ is given by \cite{FAMA_MultipleAccess1}
	\begin{equation}
		\mathbb{E}\big[\mathbf{h}_{u,k}\mathbf{h}_{u,\ell}^{\mathsf H}\big]
		= \beta_u J_0\!\left(\frac{2\pi}{\lambda} |d_k - d_\ell|\right)
		\mathbf{I}_M,
	\end{equation}
	where $\beta_u$ captures large-scale fading (path-loss and shadowing) for
	user $u$.
	
	To obtain a constructive parametric model, we express $\mathbf{h}_{u,k}$ as
	\begin{equation}
		\mathbf{h}_{u,1} = \sqrt{\beta_u}\,\mathbf{x}_{u,0},
	\end{equation}
	\begin{equation}
		\mathbf{h}_{u,k} = \sqrt{\beta_u}\!\left(\!
		\mu_k \mathbf{x}_{u,0}
		+ \sqrt{1 - \mu_k^2}\,\mathbf{x}_{u,k}\!
		\right)\!, \!\!\quad k = 2,\ldots,N.\label{channel}
	\end{equation}
 Also,  $\mathbf{x}_{u,0},\mathbf{x}_{u,2},\ldots,\mathbf{x}_{u,N}
	\sim \mathcal{CN}(\mathbf{0},\mathbf{I}_M)$ are independent, and
	$\mu_k \in [-1,1]$ controls the correlation between the reference port and
	port $k$. In other words, Port 1 is chosen as a reference, and all other  channels corresponding to different ports are written as correlated versions of this reference channel. The parameters $\mu_k$ are chosen such that
	\begin{equation}
		\mu_k = J_0\!\left(
		\frac{2\pi}{\lambda} |d_k - d_1|
		\right), \quad k = 2,\ldots,N,\label{mu1}
	\end{equation}
	which is consistent with the spatial correlation function above.
	
	The BS is assumed to acquire instantaneous CSI only for the reference
	port $k=1$ of each user, i.e., it estimates $\mathbf{h}_{u,1}$ for
	$u=1,\ldots,U$. This is consistent with practical FAMA operation, where
	the user locally probes its ports through fast switching and identifies
	the strongest port, while feeding back a single channel estimate to the
	BS.\footnote{Only one uplink reference transmission is required per user \cite{FAMA_MultipleAccess1}.} The remaining port channels
	$\{\mathbf{h}_{u,k}\}_{k>1}$ are not required to be known instantaneously
	at the BS. Instead, their statistics follow directly from the spatial
	correlation model in \eqref{channel}, which allows the per-port SIR
	distributions and the correlation expressions in Section~\ref{subsec:corr}
	to be derived without assuming instantaneous CSI at all ports. Therefore, both MRT and ZF precoders are designed solely from the reference-port
	channel vector $\mathbf{h}_{u,1}$\textcolor{black}{.}\footnote{\textcolor{black}{
			The BS is assumed to acquire instantaneous CSI only for a single reference port per user, which is used for precoder design. This assumption is consistent with practical FAMA operation, where each user employs a single RF chain and cannot provide simultaneous CSI feedback for all fluid antenna ports without incurring excessive training overhead. The remaining ports are therefore not tracked instantaneously at the BS and are instead characterized statistically through the adopted spatial correlation model.
		}}

	\subsection{Downlink Transmission and Signal Model}
		The BS simultaneously serves $U$ users over the same time-frequency resource.
	The transmit signal is
	\begin{equation}
		\mathbf{x} = \sum_{u=1}^{U} \mathbf{w}_u s_u,
	\end{equation}
	where $\mathbf{w}_u \in \mathbb{C}^{M\times 1}$ is the precoding vector for
	user $u$ and $s_u$ is the corresponding data symbol with
	$\mathbb{E}[|s_u|^2] = P_u$. The total transmit power is constrained as
	$\sum_{u=1}^{U} \|\mathbf{w}_u\|^2 P_u \leq P_{\max}$.
	
	The received signal at the $k$-th port of user $u$ is
	\begin{equation}
		z_{u,k} = \mathbf{h}_{u,k}^{\mathsf H} \mathbf{x} + \eta_{u,k}
		= \mathbf{h}_{u,k}^{\mathsf H} \mathbf{w}_u s_u
		+ \sum_{i\neq u} \mathbf{h}_{u,k}^{\mathsf H} \mathbf{w}_i s_i
		+ \eta_{u,k},
	\end{equation}
	where $\eta_{u,k} \sim \mathcal{CN}(0,\sigma^2)$ denotes the additive noise.
	The first term represents the desired signal, whereas the second term accounts
	for multiuser interference at user $u$ when its fluid antenna is connected
	to port $k$.
	
	The instantaneous signal-to-interference-plus-noise ratio (SINR) at port $k$
	of user $u$ is therefore
	\begin{equation}
		\mathrm{SINR}_{u,k}
		= \frac{
			P_u |\mathbf{h}_{u,k}^{\mathsf H}\mathbf{w}_u|^2
		}{
			\sum_{i\neq u} P_i |\mathbf{h}_{u,k}^{\mathsf H}\mathbf{w}_i|^2
			+ \sigma^2
		}.\label{SINR1}
	\end{equation}
	 \textcolor{black}{
		Although the instantaneous SINR is defined in \eqref{SINR1}, the subsequent analysis focuses on the SIR. This choice is motivated by the interference-limited nature of dense multiuser MISO downlink systems, where the aggregate multiuser interference dominates thermal noise and largely determines performance. Focusing on the SIR enables exact analytical characterization of the per-port statistics, correlation structure, and outage behavior under FAMA, whereas incorporating noise would require additional approximations that obscure these fundamental mechanisms.
	}

%	\subsection{Port Selection and Effective SIR}
%	
%	Each user is equipped with a single RF chain and can connect it to exactly
%	one of its $N$ ports at any given time. We assume perfect and instantaneous
%	knowledge of the per-port SINR at the user side, so that user $u$ selects
%	the port that maximizes its instantaneous SINR:
%	\begin{equation}
%		k_u^\star = \arg\max_{k \in \{1,\ldots,N\}}
%		\mathrm{SINR}_{u,k}.
%	\end{equation}
%	The effective SINR of user $u$ under MISO–FAMA operation is then
%	\begin{equation}
%		\Gamma_u = \max_{k=1,\ldots,N}
%		\mathrm{SINR}_{u,k}.
%	\end{equation}
%	For a target SINR threshold $\bar{\gamma}$, user $u$ is said to be in outage
%	if $\Gamma_u < \bar{\gamma}$. The corresponding outage probability is
%	\begin{equation}
%		P_{\mathrm{out},u}(\bar{\gamma})
%		= \Pr\{\Gamma_u < \bar{\gamma}\}
%		= \Pr\left\{
%		\max_{k} \mathrm{SINR}_{u,k} < \bar{\gamma}
%		\right\}.
%	\end{equation}
	
	\subsection{Beamforming Design}
	
	In this work, we consider linear precoding at the BS. The beamforming
	vectors are designed based on the CSI at a
	reference port for each user as mentioned. For analytical tractability, we assume that
	the BS acquires the channel vectors $\{\mathbf{h}_{u,1}\}_{u=1}^{U}$,
	corresponding to the first port of each user, and constructs the precoders
	accordingly. The extension to other reference ports follows by symmetry. 	Let $\mathbf{H}_1 = [\mathbf{h}_{1,1},\ldots,\mathbf{h}_{U,1}] \in
	\mathbb{C}^{M\times U}$ denote the aggregated channel matrix at the
	reference port. We consider two conventional linear precoding strategies:
 MRT and ZF. \textcolor{black}{
 	Both MRT and ZF precoders are constructed solely based on the CSI of the reference port for each user. The resulting beamformers are thus agnostic to the instantaneous channels of the remaining fluid antenna ports, and the FAMA gain arises from local port selection at the user side rather than from BS-side per-port adaptation.}\footnote{\textcolor{black}{While different choices or adaptive selection of the reference port may affect the instantaneous performance, the fixed reference-port assumption adopted here provides a representative and practically motivated baseline, and allows the fundamental impact of precoding, spatial correlation, and fluid-antenna port selection to be isolated analytically. A systematic investigation of adaptive reference-port selection and its impact on performance constitutes an interesting direction for future work. 
 }
 }
	
	\subsubsection{MRT  Precoding}
	
	Under MRT, the beamformer for user $u$ is chosen to align with its channel
	vector at the reference port, i.e., Port 1.  Specifically,
	\begin{equation}
		\mathbf{w}_u^{\mathrm{MRT}} =
		\frac{\mathbf{h}_{u,1}}{\|\mathbf{h}_{u,1}\|}, \quad u=1,\ldots,U.
	\end{equation}
	The MRT precoder maximizes the average received signal power at the
	reference port of each user, but does not actively suppress multiuser
	interference. The resulting desired signal term at port $k$ of user $u$ is
	\begin{equation}
		\mathbf{h}_{u,k}^{\mathsf H}\mathbf{w}_u^{\mathrm{MRT}}
		= \frac{\mathbf{h}_{u,k}^{\mathsf H}\mathbf{h}_{u,1}}
		{\|\mathbf{h}_{u,1}\|},
	\end{equation}
	while the interference contribution from user $i\neq u$ is
	\begin{equation}
		\mathbf{h}_{u,k}^{\mathsf H}\mathbf{w}_i^{\mathrm{MRT}}
		= \frac{\mathbf{h}_{u,k}^{\mathsf H}\mathbf{h}_{i,1}}
		{\|\mathbf{h}_{i,1}\|}.
	\end{equation}
	
	\subsubsection{ ZF Precoding}
	
	To explicitly mitigate multiuser interference, we also consider ZF precoding
	based on the reference-port channel matrix $\mathbf{H}_1$. The unnormalized
	ZF precoder is given by
	\begin{equation}
		\widetilde{\mathbf{W}}^{\mathrm{ZF}}
		= \mathbf{H}_1
		\left(\mathbf{H}_1^{\mathsf H}\mathbf{H}_1\right)^{-1},
	\end{equation}
	where the $u$-th column $\widetilde{\mathbf{w}}_u^{\mathrm{ZF}}$ is the
	beamforming direction for user $u$. To satisfy the transmit power
	constraint, we apply per-column normalization
	\begin{equation}
		\mathbf{w}_u^{\mathrm{ZF}} =
		\frac{\widetilde{\mathbf{w}}_u^{\mathrm{ZF}}}
		{\|\widetilde{\mathbf{w}}_u^{\mathrm{ZF}}\|}, \quad u=1,\ldots,U.
	\end{equation}
	By construction, ZF precoding satisfies
	\begin{equation}
		\mathbf{h}_{i,1}^{\mathsf H}\mathbf{w}_u^{\mathrm{ZF}} = 0,
		\quad \forall\, i\neq u,
	\end{equation}
	so that the multiuser interference is completely removed at the reference
	ports in the absence of CSI errors. At a generic port $k$, however, the
	channels differ from $\mathbf{h}_{u,1}$ due to the fluid-antenna displacement
	and the resulting spatial correlation, and residual interference is
	observed:
	\begin{equation}
		\mathbf{h}_{u,k}^{\mathsf H}\mathbf{w}_i^{\mathrm{ZF}} \neq 0,
		\quad k\geq 2, \; i\neq u.
	\end{equation}
This residual interference arises because the precoder is designed using CSI from a single reference port, while reception may occur at a different port.\footnote{\textcolor{black}{
		It is worth noting that imperfect CSI or channel estimation errors at the reference port would primarily modify the effective signal dimensions and the strength of spatial correlation across ports, without altering the fundamental analytical structure of the proposed per-port SIR distributions and outage analysis.
}}

\textcolor{black}{
	When a non-reference port is selected, this CSI mismatch introduces residual multiuser interference. This effect is inherent to FAMA under single-port CSI constraints and is captured by the derived per-port SIR distributions and cross-port correlation, influencing the achievable selection diversity and outage performance.
}

\textcolor{black}{
	The adopted ZF precoder relies on single-port CSI, reflecting practical CSI acquisition constraints in FAMA systems. Under this constraint, reference-port ZF provides an analytically tractable baseline for studying the interaction between multiuser interference suppression and fluid-antenna port selection.\footnote{\textcolor{black}{
			ZF perfectly nulls interference only at the reference port; selecting a different port results in residual interference due to CSI mismatch, reflecting the inherent sensitivity of ZF under FAMA's single-port CSI assumption.
	}}
}
 \textcolor{black}{
	More sophisticated precoding strategies, such as robust or statistical ZF designs that explicitly account for port correlation or uncertainty, may further mitigate residual interference across the fluid-antenna aperture. However, such designs require additional CSI or statistical knowledge and are beyond the scope of the present analysis, which focuses on establishing fundamental performance benchmarks under minimal CSI assumptions.
}

%	The objective of the subsequent analysis is to characterize the statistical
%	behavior of $\Gamma_u$ and to quantify the performance gains of MISO–FAMA
%	in terms of outage probability and rate.
	
	\subsection{Per-Port SIR}
	 For clarity, we
	focus on a fixed user $u$ and a generic port index $k$.
	
	The per-port SIR is defined  as
	\begin{equation}
		X_{u,k}
		\triangleq
		\frac{
			P_u |\mathbf{h}_{u,k}^{\mathsf H}\mathbf{w}_u|^2
		}{
			\sum_{i\neq u} P_i |\mathbf{h}_{u,k}^{\mathsf H}\mathbf{w}_i|^2
		}.
	\end{equation}
	
	 The subsequent analysis focuses on the
	statistical characterization of this per-port SIR and its maximization over
	the $N$ fluid-antenna ports. \textcolor{black}{Henceforth, equal transmit power allocation across users is assumed to provide a fair and analytically tractable baseline that isolates the effects of multiuser precoding, spatial correlation, and fluid-antenna port selection. While more general power-control strategies could be considered, they would primarily scale the effective SIR parameters without altering the analytical structure developed in this work.}

	 The SIR can then be
	written as the ratio
	\begin{equation}
		X_{u,k}
		= \frac{U_k}{V_k},
	\end{equation}
	where
	\begin{align}
		U_k &\triangleq |\mathbf{h}_{u,k}^{\mathsf H}\mathbf{w}_u|^2,
		\label{eq:Uk_def}\\
		V_k &\triangleq \sum_{i\neq u}
		|\mathbf{h}_{u,k}^{\mathsf H}\mathbf{w}_i|^2.
		\label{eq:Vk_def}
	\end{align}
	
		\textcolor{black}{
		The per-port SIR derivations exploit conditional distributional properties under i.i.d.\ Rayleigh fading and reference-port CSI. For a generic port $k$, the channel can be written in the correlated form
		$\mathbf{h}_{u,k}=\mu_k \mathbf{h}_{u,1}+\sqrt{1-\mu_k^2}\,\mathbf{e}_{u,k}$,
		where $\mathbf{e}_{u,k}\sim\mathcal{CN}(\mathbf{0},\mathbf{I})$ is independent of $\mathbf{h}_{u,1}$.
		Conditioned on the reference-port CSI used to construct the precoders, the beamforming vectors are deterministic. 			
		For ZF precoding, the interfering beamformers satisfy
		$\mathbf{h}_{u,1}^{\mathsf H}\mathbf{w}_i^{\mathrm{ZF}}=0$ for all $i\neq u$.
		Hence, the correlated component $\mu_k \mathbf{h}_{u,1}$ does not contribute to the interference at port $k$, and the aggregate interference depends only on the innovation component $\mathbf{e}_{u,k}$, which is independent of the desired projection term. 		
		For MRT precoding, the beamformer of user $u$ depends only on $\mathbf{h}_{u,1}$, while the interfering beamformers depend on $\mathbf{h}_{i,1}$ for $i\neq u$.
		Due to the independence of user channels and the rotational invariance of Rayleigh fading, the desired projection $\mathbf{h}_{u,k}^{\mathsf H}\mathbf{w}_u^{\mathrm{MRT}}$ and the interference projections $\mathbf{h}_{u,k}^{\mathsf H}\mathbf{w}_i^{\mathrm{MRT}}$ ($i\neq u$) are independent in distribution when conditioned on the beamformers.
		This conditional independence enables the closed-form Beta-prime characterization of the per-port SIR under MRT.}
%		\subsection{Single-Port Outage and FAMA Over $N$ Ports}
%	
%We begin by characterizing the outage probability of the per-port SIR~$X$, which
%serves as the fundamental building block for the overall FAMA performance. Let
%$\gamma>0$ denote the target SIR threshold. The single-port outage probability
%$\Pr\{X_{u,k}<\gamma\}$ quantifies the reliability at an individual fluid-antenna
%position. Building on this result, we then analyze FAMA operation over $N$ ports,
%where the receiver selects the port with the maximum SIR. This leads to the
%FAMA outage probability $\Pr\{\max_{k=1,\dots,N} X_{u,k} < \gamma\}$, which depends
%jointly on the per-port SIR distribution and the spatial dependence across ports.

	\section{Statistical SIR Analysis}\label{statistical}
	In this section, we characterize the per-port SIR statistics under specific linear precoders, which are MRT and ZF. \textcolor{black}{To aid readability, we briefly outline the main steps of the derivation before presenting the final expression. The key idea is to condition on the desired-signal power, exploit the independence of the interference terms given the precoder structure, and then marginalize over the resulting distributions.
	}
	
	\subsection{MRT-Based Statistical Model}
		To obtain a tractable analytical form, we first focus on MRT precoding
	based on the reference-port CSI. 
	Under the i.i.d. Rayleigh assumption across users and the fluid-antenna
	correlation model in Section~\ref{system}, the channel vectors remain circularly
	symmetric complex Gaussian, and the beamforming directions
	$\{\mathbf{w}_i\}_{i=1}^{U}$ are isotropically distributed on the complex
	unit sphere.\footnote{\textcolor{black}{
The analytical tractability of the proposed framework relies on the rotational invariance of Rayleigh fading, which ensures isotropically distributed channel vectors and enables closed-form characterization of projected gains under MRT and ZF, as well as tractable cross-port correlation modeling. In Rician fading, the deterministic LoS component breaks isotropy and introduces direction-dependent statistics, so the projected gain distributions depend on the relative orientation of the LoS component, precoding subspace, and antenna displacement. Consequently, the exact analysis does not extend directly to Rician channels. Extending the framework to Rician fading remains an important direction for future work.
		}
	} For a fixed port $k$ and user $u$, it is convenient to write
	\begin{equation}
		\mathbf{h}_{u,k}
		= \sqrt{\beta_u}\,\mathbf{g}_{u,k},
	\end{equation}
	where $\mathbf{g}_{u,k}\sim\mathcal{CN}(\mathbf{0},\mathbf{I}_M)$ is
	standard complex Gaussian. Using the rotational invariance of
	$\mathbf{g}_{u,k}$, it can be shown that (for a given $k$) \cite{Tse2005}
	\begin{equation}
		U_k^{\mathrm{MRT}} = |\mathbf{h}_{u,k}^{\mathsf H}\mathbf{w}_u|^2
		\sim \Gamma(M,1),\label{mrt1}
	\end{equation}
	i.e., $U_k$ is Gamma distributed with shape $M$ and unit scale.\footnote{\textcolor{black}{
			A rigorous derivation of \eqref{mrt1} is provided in Appendix~A. In particular, we condition on the MRT beamformer constructed from the reference-port channel, use the isotropy of Rayleigh fading to characterize the projection of $\mathbf{h}_{u,1}$ onto the normalized beam direction, and account for the beamformer normalization to obtain the resulting Gamma-distributed effective channel gain.
		}
	} Moreover,
	each interference term in \eqref{eq:Vk_def} satisfies
	\begin{equation}
		|\mathbf{h}_{u,k}^{\mathsf H}\mathbf{w}_i|^2 \sim \mathrm{Exp}(1),
		\quad i\neq u,
	\end{equation}
	and the summation of $L=U-1$ independent exponential terms yields
	\begin{equation}
		V_k^{\mathrm{MRT}} \sim \Gamma(L,1),
	\end{equation}
	where $L$ is the number of interfering data streams. The above relations
	hold marginally for each fixed port $k$. Notably, these expressions capture the per-port power
	statistics of the desired and interfering signals under MRT and Rayleigh
	fading.
	
	\subsubsection{SIR Distribution  Under MRT}
	
	Combining \eqref{eq:Uk_def}–\eqref{eq:Vk_def} with the Gamma distributions
	of $U_k$ and $V_k$, the per-port SIR $	X_{u,k}^{\mathrm{MRT}} = \frac{U_k^{\mathrm{MRT}}}{V_k^{\mathrm{MRT}}}$ 
	follows a Beta-prime distribution with parameters $(M,L)$ \cite{Papoulis1965,Simon2004}.\footnote{A Beta-prime, also known as Pearson Type~VI  distribution
		with parameters $(a,b)$ \cite{Simon2004} has a pdf  given as
		\begin{equation}
			f_X(x)
			=
			\frac{x^{a-1}(1+x)^{-(a+b)}}{B(a,b)},
			\qquad x>0,
			\label{eq:beta_prime_pdf}
		\end{equation}
		where $B(a,b)=\Gamma(a)\Gamma(b)/\Gamma(a+b)$ denotes the Beta function.
		The cumulative distribution function (CDF) is given in terms of the
		regularized incomplete Beta function $I_y(a,b)$:
		\begin{equation}
			F_X(x)
			= \Pr\{X<x\}
			= I_{\frac{x}{1+x}}(a,b),
			\qquad x>0,
			\label{eq:beta_prime_cdf}
		\end{equation}
		where
		\begin{equation}
			I_y(a,b)
			= \frac{1}{B(a,b)}
			\int_{0}^{y} t^{a-1}(1-t)^{b-1}\,dt,
			\qquad 0<y<1.\label{beta}
	\end{equation}} The cumulative
	distribution function (CDF) of $X_{u,k}$ at threshold $\gamma>0$ is given by
	\begin{equation}
		F_X(\gamma)
		\triangleq \Pr\{X_{u,k} < \gamma\}
		= I_{\frac{\gamma}{1+\gamma}}(M,L).
		\label{eq:FX_gamma_beta}
	\end{equation}
 For
	integer $M$ and $L$, \eqref{eq:FX_gamma_beta} admits the closed-form
	finite-sum representation
	\begin{equation}
		F_X(\gamma)
		=
		1 - \frac{1}{(1+\gamma)^{M+L-1}}
		\sum_{j=0}^{M-1}
		\binom{M+L-1}{j}\gamma^{j},
		\label{eq:FX_gamma_finitesum}
	\end{equation}
	since the incomplete beta function has the finite-sum expansion \cite[Eq. 6.6.4]{Abramowitz1964}
	\begin{align}
		I_{y}(a,b)
		= 
		1 - 
		\sum_{j=0}^{a-1}
		\binom{a+b-1}{\,j\,}
		y^{\,j}(1-y)^{\,a+b-1-j}.
	\end{align}
	The expressions in \eqref{eq:FX_gamma_beta}, \eqref{eq:FX_gamma_finitesum}
	fully characterize the SIR statistics at an individual fluid-antenna port
	under MISO downlink transmission with MRT precoding. 
	
	\subsection{ZF-Based Statistical Model}
	
	We now characterize the per-port SIR statistics under  ZF
	precoding based on the reference-port CSI. As in the MRT case, the BS
	acquires the channel vectors $\{h_{u,1}\}_{u=1}^{U}$ corresponding to the
	first port of each user and constructs the ZF precoder accordingly.
	
	As in the MRT case, we operate in the interference-limited regime and define
	the desired-signal and interference powers at port $k$ as
	\begin{align}
		U^{\mathrm{ZF}}_{k}
		&\triangleq
		\big|\mathbf{h}_{u,k}^{H}\mathbf{w}_{u}^{\mathrm{ZF}}\big|^{2},
		\label{eq:UkZF} \\
		V^{\mathrm{ZF}}_{k}
		&\triangleq
		\sum_{i\neq u}
		\big|
		\mathbf{h}_{u,k}^{H}\mathbf{w}_{i}^{\mathrm{ZF}}
		\big|^{2}
		\label{eq:VkZF}.
	\end{align}
	
	The per-port SIR under ZF is 
	\begin{equation}
		X_{u,k}^{\mathrm{ZF}}
		= \frac{U^{\mathrm{ZF}}_{k}}{V^{\mathrm{ZF}}_{k}}.
	\end{equation}

At the reference port $k=1$, the ZF beamforming vector
$\mathbf{w}_{u}^{\mathrm{ZF}}$ is orthogonal to the co-user channel
vectors $\{\mathbf{h}_{i,1}\}_{i\neq u}$ and lies in the
$(M-U+1)$–dimensional nullspace of $\mathbf{H}_{1}$.  
Since $\mathbf{h}_{u,1}$ is independent of this nullspace and is
circularly symmetric complex Gaussian, its projection onto the
ZF beamforming direction has $(M-U+1)$ effective degrees of
freedom. Therefore,
\begin{equation}
	U^{\mathrm{ZF}}_{1}
	= \big|\mathbf{h}_{u,1}^{H}\mathbf{w}_{u}^{\mathrm{ZF}}\big|^{2}
	\sim \Gamma(M-U+1, 1),
	\label{eq:UZF_ref_no_proj}
\end{equation}
which follows directly from the standard ZF property that the
useful signal under ZF beamforming experiences a chi-square
distribution with $2(M-U+1)$ degrees of freedom.\footnote{\textcolor{black}{
		A rigorous derivation of \eqref{eq:UZF_ref_no_proj} is given in Appendix~A. ZF projects the reference-port channel of user $u$ onto the nullspace of the other users' reference-port channels, whose dimension is $M-U+1$. By standard results on projections of complex Gaussian vectors, the resulting effective gain is Gamma distributed with shape parameter $M-U+1$ after normalization.
	}
} For a general port $k\ge 2$, $\mathbf{h}_{u,k}$ remains circularly
	symmetric Gaussian and is independent of the ZF nullspace orientation.
	Thus, by rotational invariance, we obtain
	\begin{equation}
		U^{\mathrm{ZF}}_{k}
		\sim \Gamma(M-U+1,1),
		\qquad \forall k.
		\label{eq:UkZFgeneral}
	\end{equation}

For a fixed user $u$ and port $k$, the vectors $\mathbf{h}_{u,k}$ and
$\mathbf{w}_{i}^{\mathrm{ZF}}$ (for $i\neq u$) are independent and isotropic
in $\mathbb{C}^{M}$, which implies the well-known scalar-product result
\begin{equation}
	\big|\mathbf{h}_{u,k}^{H}
	\mathbf{w}_{i}^{\mathrm{ZF}}\big|^{2}
	\sim \exp(1),
	\qquad i\neq u.
	\label{eq:ZF_exp}
\end{equation}
Since $L=U-1$ such terms contribute independently. Thus, we have
\begin{equation}
	V^{\mathrm{ZF}}_{k}
	\sim \Gamma(L,1),
	\qquad \forall k.
	\label{eq:VkZFfinal}
\end{equation}
	
	\subsubsection{SIR Distribution Under ZF}
	Since $U^{\mathrm{ZF}}_{k}\sim\Gamma(M-U+1,1)$ and
	$V^{\mathrm{ZF}}_{k}\sim\Gamma(L,1)$ and these are independent for any fixed
	port $k$, the ratio $X_{u,k}^{\mathrm{ZF}}$
	follows a Beta-prime distribution with parameters $(M-U+1,L)$. Its CDF is
	\begin{equation}
		F_{X}^{\mathrm{ZF}}(\gamma)
		=
		I_{\frac{\gamma}{1+\gamma}}(M-U+1,L),
		\label{eq:ZF_betaprime}
	\end{equation}
	showing explicitly that ZF reduces the desired-signal shape parameter from
	$M$ (MRT) to $(M-U+1)$ due to nulling of co-user interference.
	
	Analogous to the MRT case in~\eqref{eq:FX_gamma_finitesum}, this CDF also
	admits a closed-form finite-sum representation for integer parameters.
	Specifically, we have
	\begin{equation}
		F_{X}^{\mathrm{ZF}}(\gamma)
		= 1
		- \frac{1}{(1+\gamma)^{M-U+L}}
		\sum_{j=0}^{M-U}
		\binom{M-U+L}{j}\gamma^{j}.
		\label{eq:FX_ZF_finitesum}
	\end{equation}
	Similar to \eqref{eq:ZF_betaprime}, this expression highlights the reduced effective array gain under ZF,
	reflected in the smaller shape parameter $(M-U+1)$ compared to the MRT
	value $M$. In the subsequent
	sections, the  per-port statistics for MRT and ZF are used to analyze the outage
	probability of the FAMA-enabled MISO system, where each user selects the
	port that maximizes its instantaneous SIR.
	
\begin{remark}[MRT vs.\ ZF: Comparison of Per-Port SIR Distributions]
	Both MRT and ZF lead to per-port SIRs that follow a Beta-prime distribution,
	differing only in the effective numerator shape parameter. MRT attains the full
	array gain of $M$ antennas and therefore yields a larger shape parameter, while
	ZF suppresses multiuser interference at the reference port but incurs an
	array-gain reduction to $M-U+1$. Under FAMA port selection, the distinction
	between the two precoders is thus reflected directly in the corresponding
	Beta-prime parameters $(M,L)$ for MRT and $(M-U+1,L)$ for ZF. As shown in
	Section~\ref{subsec:corr}, the cross-port SIR correlation $\rho_X(k,\ell)$
	decreases with the effective array dimension. Consequently, MRT induces weaker
	correlation across ports than ZF, enabling stronger selection diversity and 
	improved FAMA performance.
\end{remark}

%	\subsection{Order-Statistics Approximation for Correlated Ports}
%	
%Due to the correlation among $\{X_{u,k}\}$, the exact CDF of the maximum
%$\Gamma_u = \max_{k} X_{u,k}$ is analytically intractable. A widely used
%approach in the analysis of correlated-branch selection diversity is to
%approximate the $N$ correlated branches by an equivalent number of
%independent branches \cite{SelDivCorr1,SelDivCorr2}. In this framework, the
%effective number of ports is modeled as
%\begin{equation}
%	N_{\mathrm{eff}}
%	=
%	\frac{N}{1 + (N-1)\rho_X},
%\end{equation}
%where $\rho_X$ denotes the average of the per-pair correlation coefficients
%in \eqref{eq:rhoX_final_PDF}. Substituting $N_{\mathrm{eff}}$ into the
%independent-branch expression yields the approximation
%\begin{equation}
%	P_{\mathrm{out}}(\gamma)
%	\approx
%	\big[F_X(\gamma)\big]^{N_{\mathrm{eff}}},
%\end{equation}
%which recovers the ideal i.i.d.\ behavior when $\rho_X\to 0$, and reduces
%to single-port performance when $\rho_X\to 1$.
%
%Since MRT provides a larger effective signal dimension
%($M_{\mathrm{eff}}=M$) than ZF ($M_{\mathrm{eff}}=M-U+1$), it induces weaker
%cross-port correlation and therefore a larger $N_{\mathrm{eff}}$, rendering
%FAMA more effective under MRT for a given port geometry.

	\section{FAMA Outage Probability Analysis}\label{fama}
In this section, we characterize the outage performance of FAMA by progressively analyzing the
per-port SIR distribution, the effect of selection over multiple ports, and the role of spatial
correlation across the aperture. Building on these components, we derive tight analytical upper
and lower bounds on the FAMA outage and establish its small- and large-SIR asymptotic behavior.

		\subsection{Single-Port Outage Probability}

From \eqref{eq:FX_gamma_beta}, the outage probability at a single port $k$
of user $u$ is
\begin{equation}
	P_{\mathrm{out}}^{(1)}(\gamma)
	\triangleq
	\Pr\{X_{u,k} < \gamma\}
	= F_X(\gamma)
	= I_{\frac{\gamma}{1+\gamma}}(M_{\mathrm{eff}},L),
	\label{eq:single_port_outage}
\end{equation}
where $M_{\mathrm{eff}}$ denotes the effective signal dimension determined by
the precoder, i.e.,
\[
M_{\mathrm{eff}} =
\begin{cases}
	M, & \text{for MRT}, \\
	M-U+1, & \text{for ZF},
\end{cases}
\]
and $L$ is the number of interfering data streams. \textcolor{black}{	This result shows that the effect of multiuser precoding enters the per-port SIR distribution solely through the effective signal dimension, while the fluid-antenna displacement governs the correlation across ports.}
The finite-sum closed-form expressions in
\eqref{eq:FX_gamma_finitesum} and \eqref{eq:FX_ZF_finitesum} apply directly for integer
$(M_{\mathrm{eff}},L)$.
\textcolor{black}{
	The single-port SIR distribution and outage probability are  not intended to represent the performance of a FAMA system, but rather to serve as exact analytical building blocks for the subsequent characterization of multiport selection under correlation.
}

	\subsection{FAMA Selection Over $N$ Ports}
	
	Under FAMA operation, user $u$ selects the port with the strongest
	instantaneous SIR. The effective SIR is therefore
	\begin{equation}
		\Gamma_u = \max_{k=1,\ldots,N} X_{u,k}.
	\end{equation}
	
	The corresponding outage probability is
	\begin{equation}
	\!\!\!	P_{\mathrm{out}}(\gamma)
	\!	=\! \Pr\{\Gamma_u < \gamma\}
	\!	=\! \Pr\{X_{u,1}<\gamma,\ldots,X_{u,N}<\gamma\}.
		\label{eq:Pout_max_def}
	\end{equation}
	
	An exact closed-form evaluation of \eqref{eq:Pout_max_def} is impeded by the
	correlation among the SIRs $\{X_{u,k}\}$ across the fluid ports, which
	arises from the shared MISO channel vectors and the spatially correlated
	fading model in Section~\ref{system}. Nevertheless, the per-port CDF $F_X(\gamma)$
	enables the derivation of rigorous and tractable bounds on
	$P_{\mathrm{out}}(\gamma)$.
	
	\subsection{Cross-Port SIR Correlation Analysis}\label{corr1}\label{subsec:corr}

	The outage performance of FAMA depends critically on the statistical
	dependence among the per-port SIRs $\{X_{u,k}\}$.
	Although each $X_{u,k}$ follows a Beta-prime distribution, the SIRs are
	not independent due to the common component in the  user-channel
	model \eqref{channel}, 
	where $\mathbf{x}_{u,0}\sim\mathcal{CN}(\mathbf{0},\mathbf{I})$ is shared
	across ports, while $\mathbf{x}_{u,k}$ are mutually independent.\footnote{	\textcolor{black}{
			For fluid-antenna systems with correlated ports, an exact closed-form expression for the outage probability under port selection would require the joint distribution of the per-port SIRs. Due to the strong statistical dependence induced by shared channel components and multiuser interference, this joint distribution does not admit a tractable analytical form, even for a single user.
	}}

	\begin{lemma}\label{lem01}
		The correlation coefficient between ports $k,\ell$ admits the approximation
		\begin{equation}
			\rho_{U}(k,\ell)
			= \mathrm{corr}(U_k,U_\ell)
			\approx
			\mu_k^{2}\mu_\ell^{2}
			\frac{M_{\mathrm{eff}}}{M_{\mathrm{eff}}+L},
			\label{eq:rhoU_final}
		\end{equation}
		where $M_{\mathrm{eff}} = M$ for MRT and $M_{\mathrm{eff}} = M-U+1$ for ZF.
			\end{lemma}
		\begin{proof}
		Please see Appendix~\ref{lem01proof}.	
		\end{proof}
	\textcolor{black}{
The proposed correlation expression is obtained via a linearization-based approximation and is accurate when the aggregate interference exhibits limited relative fluctuations, such as in moderate-to-large effective interference dimensions or non-negligible fluid-antenna apertures. In these regimes, the approximation closely matches MC estimates, as shown in Section~V.	}
	\textcolor{black}{
In strongly correlated regimes (e.g., very small apertures), quantitative accuracy may degrade. However, the achievable FAMA gain is inherently limited and the cross-port SIR correlation approaches unity, so the expression correctly captures the limiting behavior.
	}
	\textcolor{black}{
The derived cross-port SIR correlation coefficient therefore serves as a quantitative indicator of the operating regime, measuring proximity to either the fully correlated or independent-port extremes.
		}
		
	The following lemma gives an approximation of the cross-port SIR correlation.
		\begin{lemma}\label{lem02}
			The cross-port SIR correlation can be approximated as
		\begin{equation}
			\rho_{X}(k,\ell)
			\approx
					\rho_{U}(k,\ell)\,\textcolor{black}{\frac{L}{L + M_{\mathrm{eff}}+1}.}
			\label{eq:rhoX_final}
		\end{equation} 
	\end{lemma}
	\begin{proof}
		Please see Appendix~\ref{lem02proof}.	
	\end{proof}
	Equation~\eqref{eq:rhoU_final} shows that the cross-port dependence is governed
	jointly by the spatial overlap $\mu_k\mu_\ell$ and the effective array
	dimension. Also, Equation~\eqref{eq:rhoX_final} shows that the SIR correlation inherits the form of
	$\rho_{U}(k,\ell)$ but is attenuated by a factor depending solely on $L$ and $ M_{\mathrm{eff}}$.
	\textcolor{black}{
		Although the per-port SIRs are identically distributed in a marginal sense under the adopted fading and precoding models, they are not identically experienced across ports due to spatially varying correlation. Ports closer to the reference port used for precoder design exhibit stronger correlation with it, while more distant ports experience weaker correlation. This non-identical behavior is captured through the cross-port correlation coefficients and affects the joint SIR statistics and outage performance under FAMA, even though the marginal per-port distributions remain the same.
	}
	
	\paragraph*{Implications for FAMA}
	The quantity $\rho_{X}(k,\ell)$ determines the diversity that FAMA can
	extract from the $N$ ports.  
	When ports are weakly correlated ($\rho_X\!\ll\!1$), selection over $N$
	ports approaches ideal selection diversity.
	When ports are strongly correlated ($\rho_X\!\approx\!1$), the effective
	diversity reduces to that of a single port.
	Moreover, since $M_{\mathrm{eff}}$ is larger for MRT than ZF,
	\eqref{eq:rhoU_final}-\eqref{eq:rhoX_final} imply
	\[
	\rho_X^{\mathrm{MRT}} < \rho_X^{\mathrm{ZF}},
	\]
	showing that MRT naturally produces weaker port dependence than ZF for
	the same antenna geometry, thereby enhancing FAMA effectiveness.

The dependence, quantified
	by the correlation coefficient $\rho_X(k,\ell)$,
	determines how close the FAMA performance is to the ideal
	independent-port case. The subsequent bounds  apply
	without assuming independence and naturally incorporate the
	effects of cross-port correlation.
	
	\textcolor{black}{
		The correlation expressions derived in this section explicitly capture the dominant source of multiport coupling in FAMA systems, namely the common channel component shared across displaced antenna ports. This coupling is further shaped by the effective signal dimension imposed by the adopted linear precoder, leading to fundamentally different correlation behaviors under MRT and ZF. As such, the proposed correlation model provides direct physical insight into how multiuser precoding and antenna displacement jointly affect selection diversity.
	}
	\textcolor{black}{
		Moreover, the correlation coefficient  provides a quantitative indicator of where the true outage probability lies between the two bounds, enabling a correlation-aware approximation of FAMA performance without requiring the intractable exact joint SIR distribution.
	}

	\subsection{Upper and Lower Bounds on the FAMA Outage}
	
	Regarding the upper bound, we present the following theorem.
		\begin{theorem}[Upper Bound on the FAMA Outage Probability]\label{Prop1}
		Let $\{X_{u,k}\}_{k=1}^{N}$ denote the per-port SIRs of user~$u$ and $\Gamma_u$ 
		be the effective SIR under FAMA port selection. For a target SIR threshold
		$\gamma>0$, the upper bound of the outage probability is given by
		\begin{equation}
			P_{\mathrm{out}}(\gamma)
			\le
			\min_{1\le k\le N} \Pr\{X_{u,k} < \gamma\}.
			\label{eq:Pout_upper_general}
		\end{equation}
		In particular, if $\{X_{u,k}\}$ are identically distributed with CDF
		$F_X(\gamma)=\Pr\{X_{u,k}<\gamma\}$, the outage probability satisfies
		\begin{equation}
			P_{\mathrm{out}}(\gamma)
			\le F_X(\gamma).
			\label{eq:Pout_upper_iid}
		\end{equation}
	\end{theorem}
		\begin{proof}
		Please see Appendix~\ref{Prop1proof}.	
	\end{proof}

	A complementary lower bound follows from the union bound.\footnote{\textcolor{black}{
The outage bounds are introduced to provide a tractable and physically interpretable characterization of FAMA performance under spatial correlation, serving as an alternative to an exact but analytically intractable correlated-outage expression.
		}
	}
	
	\begin{theorem}[Lower Bound on the FAMA Outage Probability]\label{Prop2}
		Let $\{X_{u,k}\}_{k=1}^{N}$ and $\Gamma_u = \max_{k} X_{u,k}$ be defined as in
		Theorem~1. For a target SIR threshold $\gamma>0$, the outage probability 		satisfies
		\begin{equation}
			P_{\mathrm{out}}(\gamma)
			\ge
						1 - \sum_{k=1}^{N}
			\big(1 - \Pr\{X_{u,k} < \gamma\}\big).
			\label{eq:Pout_lower_general}
		\end{equation}
		In particular, if $\{X_{u,k}\}$ are identically distributed with CDF
		$F_X(\gamma)=\Pr\{X_{u,k}<\gamma\}$, then
		\begin{equation}
			P_{\mathrm{out}}(\gamma)
			\ge
						1 - N\big(1 - F_X(\gamma)\big).		
			\label{eq:Pout_lower_iid}
		\end{equation}
	\end{theorem}
	\begin{proof}
		Please see Appendix~\ref{Prop2proof}.	
	\end{proof}
	
	The above theorems lead to the following corollary.
	\begin{corollary}
		For any SIR threshold $\gamma>0$ and identically distributed $\{X_{u,k}\}$, the FAMA outage probability $	P_{\mathrm{out}}(\gamma)$ is bounded as
		\begin{equation}
			1 - N\!\left(1 - I_{\frac{\gamma}{1+\gamma}}(a,b)\right)
			\;\le\;
			P_{\mathrm{out}}(\gamma)
			\;\le\;
			I_{\frac{\gamma}{1+\gamma}}(a,b),
			\label{eq:Pout_bounds_general}
		\end{equation}
		where $I_{\frac{\gamma}{1+\gamma}}(a,b)$ is the per-port SIR CDF under
		linear precoding. The parameters $(a,b)$ depend on the precoder as follows:
		\begin{equation*}
			(a,b)=
			\begin{cases}
				(M,\,L), & \text{MRT},\\[1mm]
				(M-U+1,\,L), & \text{ZF}.
			\end{cases}
		\end{equation*}
	\end{corollary}

\textcolor{black}{
	This result shows that multiuser precoding affects the per-port SIR distribution solely through the effective signal dimension, while fluid-antenna displacement governs cross-port correlation.
}\textcolor{black}{
	The proposed outage bounds depend only on the marginal per-port SIR distribution and the number of ports, and are therefore valid for arbitrary cross-port correlation.
}These bounds  are tight in practically relevant regimes where spatial correlation is weak (e.g., large fluid-antenna aperture) or where the number of ports $N$ is sufficiently large so that the SIRs behave nearly independently. In these cases, the gap between the bounds is governed by the cross-port SIR correlation $\rho_X(k,\ell)$ derived in Section~\ref{subsec:corr}: when $\rho_X(k,\ell)\approx 0$, the lower bound approaches the true outage, whereas when $\rho_X(k,\ell)\approx 1$, the upper bound becomes tight. \textcolor{black}{Notably, the derived outage bounds represent structural limits associated with the two extreme multiport coupling regimes: fully correlated ports, where FAMA reduces to single-port operation, and independent ports, where ideal selection diversity is attained. Practical FAMA systems operate between these limits, with performance governed by the degree of spatial correlation induced by antenna displacement.
}

	\begin{remark}[Interpretation]
		The upper bound in~\eqref{eq:Pout_bounds_general} coincides with the
		single-port outage $F_X(\gamma)$ and thus corresponds to a pessimistic
		regime where the spatial diversity of the fluid antenna is not exploited.
		The lower bound represents the opposite extreme of (almost) independent
		ports, where selection achieves a diversity order of $N$ and the outage
		decays roughly as $[F_X(\gamma)]^{N}$. The actual FAMA outage lies between
		these two extremes and is determined by the degree of cross-port
		correlation imposed by the aperture and the propagation environment.
	\end{remark}
	
	\begin{remark}[Independent-Port Benchmark]
		The expression $[F_X(\gamma)]^{N}$ does not follow from Theorems~\ref{Prop1}, \ref{Prop2}.
		Instead, it corresponds to the idealized scenario in which the per-port
		SIRs $\{X_{u,k}\}$ are statistically independent. In that case,
		\[
		\Pr\{\Gamma_u < \gamma\}
		= \Pr\{X_{u,1}<\gamma,\ldots,X_{u,N}<\gamma\}
		= (F_X(\gamma))^{N}.
		\]
		This benchmark represents the maximum achievable selection-diversity gain
		of FAMA. Theorems~1-2 instead provide upper and lower bounds that hold
		without independence assumptions.
	\end{remark}
	
	\subsection{Asymptotic Scaling Laws}
	
	The bounds in \eqref{eq:Pout_bounds_general} enable useful insights into the
	asymptotic behavior of the FAMA outage probability under large numbers of
	ports $N$, large antenna arrays $M$, or small SIR thresholds $\gamma$. We
	briefly summarize the key scaling results.
	\textcolor{black}{
		The asymptotic analysis presented in this subsection is not intended to replace the exact finite-dimensional characterization, but to expose fundamental scaling laws and diversity mechanisms that are otherwise obscured by finite-SNR expressions. In particular, the asymptotic results reveal how the number of antennas, users, and fluid-antenna ports jointly determine the effective diversity order under multiuser and multiport coupling.
	}
	
	\subsubsection{Large-\texorpdfstring{$N$}{N} Regime}
	
	As $N \to \infty$, the diversity gain from port selection improves the
	effective SIR. Using the lower bound in
	\eqref{eq:Pout_bounds_general}, we obtain
the following proposition.
	\begin{proposition}	\label{lem1}	
		In the regime where $N\,\varepsilon(\gamma)\ll 1$ with $
			\varepsilon(\gamma) \triangleq 1 - F_X(\gamma)$, the outage probability
		admits the approximation
		\begin{align}
			P_{\mathrm{out}}(\gamma)
			&\approx \exp\big(-N\varepsilon(\gamma)\big)\nn\\
		&	=
			\exp\big(-N\big(1 - F_X(\gamma)\big)\big).
			\label{eq:proposition_exp_approx}
		\end{align}
	\end{proposition}
	\begin{proof}
		Please see Appendix~\ref{lem1proof}.	
	\end{proof}
	This asymptotic form is accurate when the ports are weakly
	correlated, i.e., when $\rho_X(k,\ell)$ is small due to a large
	fluid-antenna aperture or sufficiently separated port locations.  In other words, when the correlation is weak
	(e.g., for large fluid-antenna aperture), the behavior approaches that of
	independent selection combining, achieving diversity order $N$.
	
	\begin{corollary}[Large-$N$ FAMA Outage under MRT]
		In the case of  MRT precoding,  Proposition~\ref{lem1}
		implies the large-$N$ approximation
		\begin{equation}
			P_{\mathrm{out}}(\gamma)
			\approx
						\exp\!\Big(
			-N\big(1-I_{\frac{\gamma}{1+\gamma}}(M,L)\big)
			\Big).
		\end{equation}
			\end{corollary}
	
	\begin{corollary}[Large-$N$ FAMA Outage under ZF]
		In the case of ZF precoding,  Proposition~\ref{lem1}
		yields 
		\begin{equation}
			P_{\mathrm{out}}(\gamma)
			\approx
					\exp\!\Big(
			-N\big(1-I_{\frac{\gamma}{1+\gamma}}(M-U+1,L)\big)
			\Big).
		\end{equation}
		\end{corollary}
		
\textcolor{black}{	The convergence toward the exponential large-\(N\) form in \eqref{eq:proposition_exp_approx} also depends indirectly on the effective array dimension \(M_{\mathrm{eff}}\). Since \(F_X(\gamma)\) and the cross-port SIR correlation both depend on \(M_{\mathrm{eff}}\), schemes with a smaller effective dimension generally approach the independent-selection regime more slowly for a fixed aperture. In particular, ZF has \(M_{\mathrm{eff}}=M-U+1\), which is smaller than the MRT value \(M_{\mathrm{eff}}=M\). Therefore, ZF typically requires either a larger number of ports \(N\) or a larger aperture to attain the same level of agreement with the exponential approximation in  \eqref{eq:proposition_exp_approx}. This is consistent with the correlation analysis, where ZF exhibits stronger cross-port dependence than MRT.
	}
	
	\subsubsection{Large-\texorpdfstring{$M$}{M}  and Small-\texorpdfstring{$\gamma$}{gamma} Threshold  Regimes under MRT}
The following proposition provides a characterization of these regimes.	
	\begin{proposition}\label{lem2}
		As $\gamma \to 0$, the per-port SIR CDF under MRT precoding $	F_X^{\mathrm{MRT}}(\gamma)$
		admits the asymptotic expansion
		\begin{equation}
			F_X^{\mathrm{MRT}}(\gamma)
			\sim
			\gamma^{M}\,
			\frac{\Gamma(L+M)}{\Gamma(L)\,\Gamma(M+1)},
			\qquad \gamma \to 0,
			\label{eq:FX_MRT_small_gamma}
		\end{equation}
		i.e., the single-port outage probability decays essentially as $\gamma^{M}$
		with a polynomial prefactor determined by $(M,L)$. Moreover, for fixed $L$
		and $\gamma<1$, as $M\to\infty$ we have
		\begin{equation}
			F_X^{\mathrm{MRT}}(\gamma)
			\sim
			\frac{M^{L-1}}{\Gamma(L)}\,\gamma^{M},
			\qquad M\to\infty,
			\label{eq:FX_MRT_large_M}
		\end{equation}
		which shows super-exponential decay in $M$.
	\end{proposition}
	\begin{proof}
		Please see Appendix~\ref{lem2proof}.	
	\end{proof}

	\begin{remark}[Interpretation of the $\gamma^{M}$ Scaling]
		Proposition~\ref{lem2} shows that the single-port outage probability
		satisfies $F_X(\gamma)\sim C\,\gamma^{M}$ for small $\gamma$, where 
		$C = \Gamma(L+M)/(\Gamma(L)\Gamma(M+1))$ is a polynomial factor determined 
		by $(M,L)$. This $\gamma^{M}$ behavior is a direct consequence of the desired
		signal power under MRT beamforming. Specifically, the effective gain 
		$U = |\mathbf{h}_{u,k}^{\mathsf H}\mathbf{w}_u|^2$ reduces to 
		$U = \|\mathbf{h}_{u,k}\|^2$ for the reference port, which is the sum of 
		$M$ independent exponential random variables. Hence, $U$ follows a 
		$\Gamma(M,1)$ distribution. The Gamma CDF satisfies
		\[
		F_U(a) = \Pr\{U < a\}
		\sim \frac{a^{M}}{\Gamma(M+1)}, \qquad a \to 0,
		\]
		so achieving $U < \gamma V$ for small $\gamma$ requires the $M$ 
		independent fading components in $U$ to be simultaneously small, yielding 
		an effective diversity order of $M$. This explains why the single-port SIR 
		exhibits a $\gamma^{M}$ outage slope.
	\end{remark}
	
	\begin{remark}[Interpretation of the Small-Threshold Asymptotics]
	 When combined with FAMA port selection, this
		$M$-fold array gain at the BS is further enhanced by the spatial diversity
		across the $N$ fluid-antenna ports, leading to an overall outage behavior
		that can approach $\gamma^{MN}$ in the ideal (weakly correlated)
		multi-port regime. Strong spatial correlation across ports, governed by the
		fluid-antenna aperture and scattering environment, reduces this effective
		diversity order but does not alter the fundamental $\gamma^{M}$ scaling of
		the single-port SIR.
	\end{remark}
	
		\begin{proposition}[Small-Threshold Asymptotics under ZF]
		As $\gamma \to 0$, the per-port SIR CDF under ZF precoding $	F_X^{\mathrm{ZF}}(\gamma)$
		admits the asymptotic expansion
		\begin{equation}
			F_X^{\mathrm{ZF}}(\gamma)
			\sim
			\gamma^{M-U+1}\,
			\frac{\Gamma(L+M-U+1)}{\Gamma(L)\,\Gamma(M-U+2)},
			\qquad \gamma \to 0,
			\label{eq:FX_ZF_small_gamma}
		\end{equation}
		i.e., the single-port outage probability decays essentially as
		$\gamma^{M-U+1}$ with a polynomial prefactor determined by $(M-U+1,L)$.
		Moreover, for fixed $L$ and $\gamma<1$, as $M\to\infty$ (with $U$ fixed)
		we have
		\begin{equation}
			F_X^{\mathrm{ZF}}(\gamma)
			\sim
			\frac{(M-U+1)^{L-1}}{\Gamma(L)}\,
			\gamma^{M-U+1},
			\qquad M\to\infty,
			\label{eq:FX_ZF_large_M}
		\end{equation}
		which shows super-exponential decay in the effective ZF dimension
		$M-U+1$.
	\end{proposition}
\begin{proof}
The proof is omitted since it follows similar lines to Appendix~\ref{lem2proof}.	
\end{proof}

\begin{corollary}[Effective Diversity Orders of MRT and ZF under FAMA]
	In the small-threshold regime $\gamma\to 0$, the single-port outage
	probabilities under MRT and ZF satisfy
	\begin{align}
		F_X^{\mathrm{MRT}}(\gamma)
		&\sim\,\gamma^{M},\\
		F_X^{\mathrm{ZF}}(\gamma)
		&\sim\,\gamma^{M-U+1}.
	\end{align} In the ideal weakly correlated regime,
	FAMA port selection over $N$ ports behaves as if it were performing
	selection combining on these per-port SIRs, leading to effective diversity
	orders
	\begin{equation}
		d_{\mathrm{MRT}} \approx MN,
		\qquad
		d_{\mathrm{ZF}} \approx (M-U+1)N,
	\end{equation}
	for MRT and ZF, respectively. Hence, ZF exhibits a diversity-order loss
	of approximately $U-1$ per port (or $(U-1)N$ overall) compared to MRT,
	reflecting the degrees of freedom spent on nulling  user interference.
\end{corollary}

\subsubsection{Large-SIR Asymptotics}

While the small-$\gamma$ behavior determines the diversity order of the system,
the large-SIR tail is equally informative in interference-limited systems.
\begin{proposition}\label{prop3}
	Let $X_{u,k}$ denote the per-port SIR of user $u$ at port $k$, and assume
	$X_{u,k}\sim\mathrm{Beta\mbox{-}prime}(a,b)$ with shape parameters
	$a>0$ and $b>0$. Then, as $\gamma\to\infty$, the upper tail of the CDF
	$F_X(\gamma)=\Pr\{X_{u,k}<\gamma\}$ satisfies
	\begin{equation}
		1 - F_X(\gamma)
		\sim
			\frac{1}{b\,B(a,b)}\,\gamma^{-b}.
		\label{eq:LargeSIR_tail_general}
	\end{equation}
	\end{proposition}
\begin{proof}
	Please see Appendix~\ref{prop3proof}.	
\end{proof}

Since the denominator shape parameter is $b=L$ for both MRT and ZF,
the tail behavior is universally
\begin{equation}
	\Pr\{X_{u,k}>\gamma\}
	\sim
	\frac{1}{b\,B(a,b)}\,\gamma^{-L},
	\qquad \gamma\to\infty.\label{large1}
\end{equation}
Thus, the downlink operates in an interference-limited regime where the
high-SIR decay rate is controlled solely by the number of interfering
users $L$, irrespective of the beamforming scheme.

\subsection{Design Insights for FAMA Systems}
The analytical results derived in Sections~\ref{statistical}-\ref{fama} yield several practical
implications for the design of MRT- and ZF-based FAMA schemes. First, the
per-port SIR statistics differ mainly through the effective array
dimension: $M_{\mathrm{eff}} = M$ for MRT and $M_{\mathrm{eff}} = M-U+1$ for
ZF. MRT therefore provides stronger low-threshold decay and weaker
cross-port correlation, making it preferable when $M \gg U$ and residual
interference dominates. ZF becomes advantageous when inter-user
interference must be aggressively suppressed and $M$ is only moderately
larger than $U$.

Second, the fluid-antenna aperture directly controls the correlation
coefficients $\rho_X(k,\ell)$ through the overlap factors $\mu_k\mu_\ell$.
Larger apertures reduce correlation and substantially increase the FAMA
selection gain. Finally, the outage bounds in
\eqref{eq:Pout_bounds_general} reveal that additional ports are most useful
when the aperture ensures weak correlation; otherwise, the diversity gain
saturates and extra ports provide diminishing returns.

\textcolor{black}{It should be noted that the selection gain obtained by increasing the number of fluid-antenna ports $N$ comes at the cost of additional port-probing overhead. Since a single-RF-chain FAMA receiver typically evaluates the candidate ports sequentially before selecting the strongest-SIR port, larger $N$ may increase the port-scanning time, processing delay, and channel-coherence requirements. Therefore, the number of ports should be chosen to balance the achievable selection diversity against the temporal overhead of practical port probing.}

\section{Numerical Results}\label{Numerical}
This section validates the analytical results derived in Sections~\ref{statistical}-\ref{fama}
using MC simulations.
\textcolor{black}{
	The numerical results further show that, for practical aperture sizes, the true outage probability closely follows a correlation-dependent trajectory between the analytical bounds, thereby demonstrating that the proposed framework accurately captures the unique performance gains enabled by FAMA.
}\textcolor{black}{
	The simulations also assess the accuracy of the proposed correlation approximations, outage bounds, and asymptotic expressions, confirming that these analytical tools track the true system behavior across varying antenna configurations, numbers of ports, aperture sizes, and spatial correlation levels.
}

\textcolor{black}{\paragraph*{Simulation Setup}
	Unless otherwise stated, all  MC simulations are obtained with $10^6$ independent channel realizations. Small-scale fading is modeled as i.i.d.\ Rayleigh fading across antennas and users, while spatial correlation across fluid-antenna ports is generated according to the displacement-based correlation model described in Section~II. The BS is equipped with $M=8$ antennas and serves $U=4$ users simultaneously using MRT or ZF precoding designed from reference-port CSI. Each user employs a fluid antenna with $N=8$ candidate ports uniformly distributed over the specified aperture. For each realization, the user selects the port that maximizes the instantaneous SIR. Unless otherwise specified, equal transmit power is allocated across users.
}

\subsection{Validation of Per-Port SIR Distributions}
Fig.~\ref{fig1} illustrates the per–port SIR CDF for both MRT and ZF precoding, for several effective array dimensions
$M_{\mathrm{eff}}$, and compares the MC results with the analytical
finite–sum expression derived in~\eqref{eq:FX_gamma_finitesum}, \eqref{eq:FX_ZF_finitesum}. As predicted by the Beta–prime model, the
per–port SIR distribution depends on the shape parameters $(a,b)$,
namely $(M,L)$ for MRT and $(M-U+1,L)$ for ZF.

For MRT, increasing $M$ shifts the CDF to the right, confirming the
$\Gamma(M,1)$ behaviour of the desired signal and the low–SIR decay
proportional to $\gamma^{M}$. In contrast, ZF exhibits a smaller
effective diversity order determined by the null–space dimension
$(M-U+1)$. Consequently, the ZF curves lie to
the left of the MRT ones. Across all shown $M_{\mathrm{eff}}$, the analytical finite–sum CDF matches
the MC simulations exactly, confirming the Beta–prime
characterisation established in Section~\ref{statistical}. 

\subsection{Cross-Port SIR Correlation Verification}
Fig.~\ref{fig3} compares the simulated SIR correlation coefficients 
 obtained from MC simulation with the analytical 
expression in~\eqref{eq:rhoX_final}. The theoretical curves accurately capture both the decay 
profile and the relative strength of the correlation across ports, confirming 
that the derived model in Section~\ref{subsec:corr} correctly characterizes the 
spatial dependence induced by the J$_0$-type channel correlation. As predicted 
by~\eqref{eq:rhoX_final}, MRT exhibits systematically weaker correlation than ZF due to its 
larger effective dimension $M_{\mathrm{eff}}$, which enhances the mixing effect 
across ports. Conversely, ZF yields higher correlation levels because its 
smaller $M_{\mathrm{eff}}$ reduces the random fluctuation of the desired and 
interference terms. Although small deviations appear at the strongest 
correlated pairs (e.g., adjacent ports), the overall agreement between analysis 
and simulation is very good, validating that the proposed correlation model 
accurately reflects both the magnitude and spatial structure of port-level SIR 
coupling. \textcolor{black}{Furthermore, Fig.~\ref{fig3} shows both a  strongly correlated case ($W = 0.25\lambda$) and a weakly correlated case ($W = 4\lambda$). This comparison demonstrates that the analytical cross-port SIR correlation model accurately captures the dependence on the aperture size and remains consistent across different spatial-correlation regimes. In particular, smaller apertures induce significantly stronger correlation across ports, while larger apertures lead to weaker coupling, thereby confirming the regime-dependent behavior that underpins the outage-bound tightness analysis in Section~IV-D.}

	\begin{figure}%
	\centering
	\includegraphics[width=0.95\linewidth]{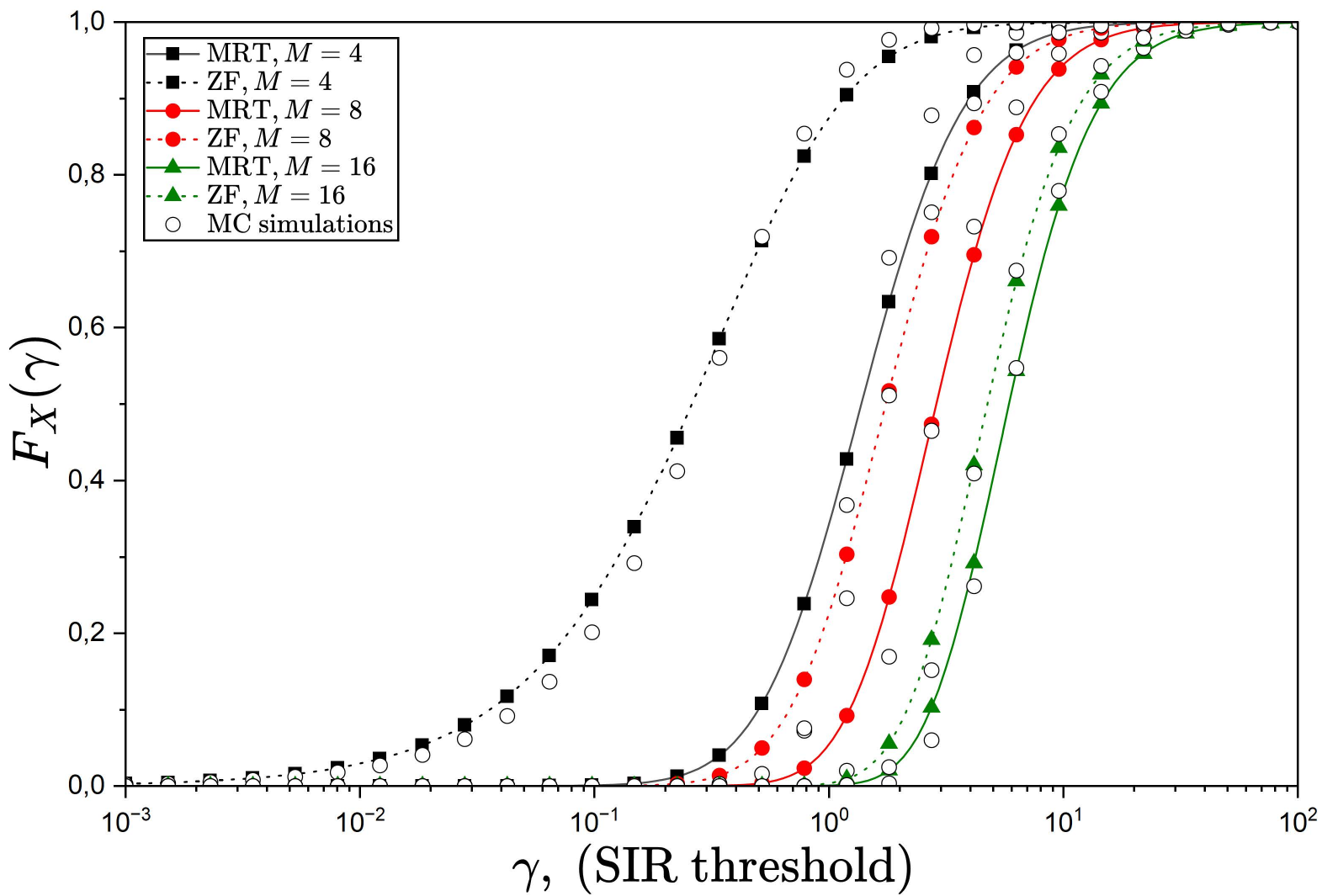}
	\caption{Per-port SIR CDF: Simulated vs. finite-sum models for varying $M$.}
	\label{fig1}
\end{figure}

	\begin{figure}%
	\centering
	\includegraphics[width=0.95\linewidth]{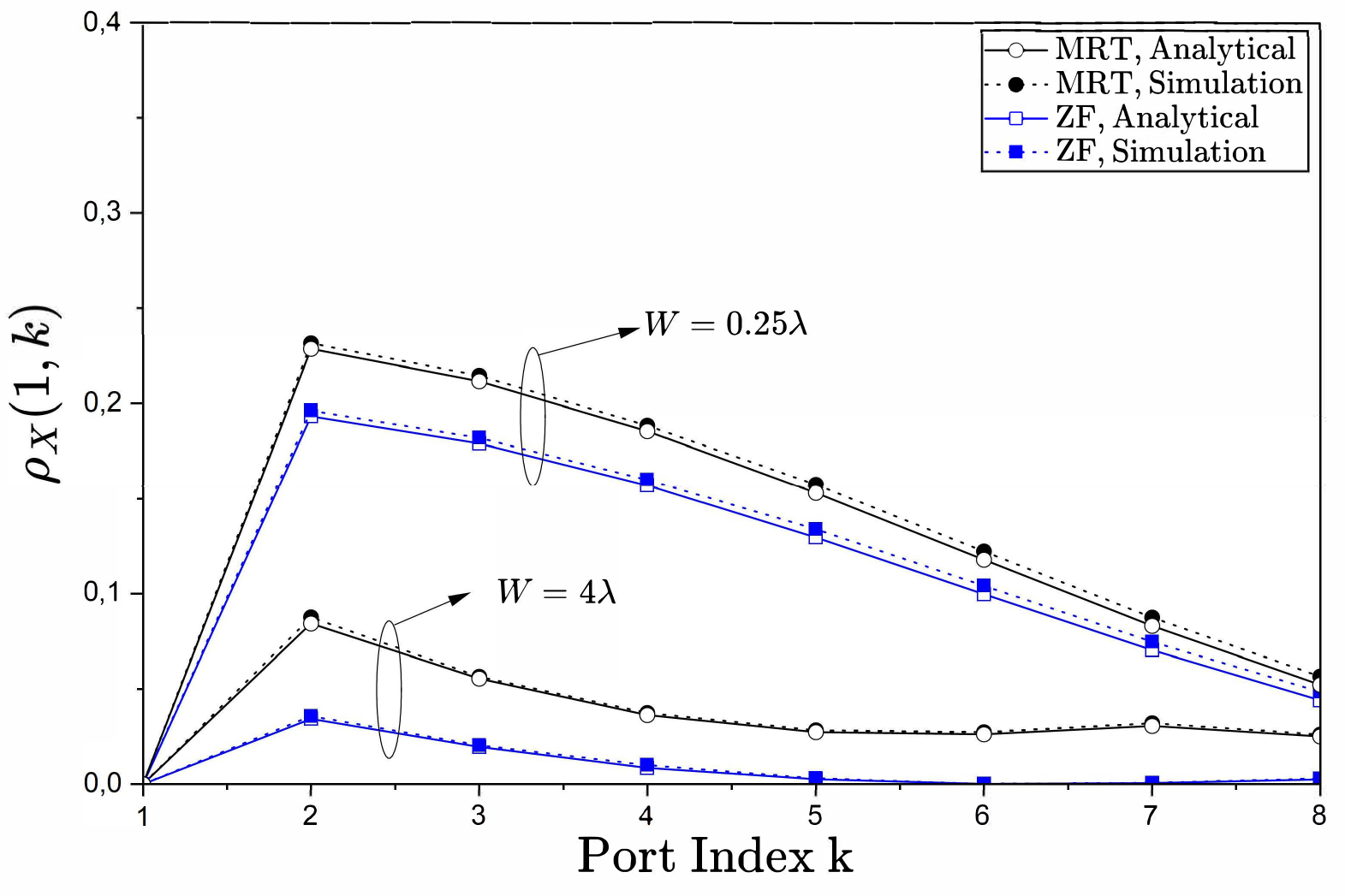}
	\caption{Simulated versus analytical SIR correlation across ports for different aperture $W$.}
	\label{fig3}
\end{figure}
\subsection{Outage Probability and Analytical Bounds}
Fig.~\ref{fig21}  presents a unified comparison between the idealized i.i.d.\ FAMA
model and the practically relevant spatially correlated FAMA scenario
for both MRT and ZF precoding. The left panel corresponds to MRT and the
right panel to ZF. The upper and lower bounds in~\eqref{eq:Pout_bounds_general} are  included and
match the exact analytical behaviour. \textcolor{black}{
	For comparison purposes, we include two baseline schemes. 
	First, a single-port transmission benchmark ($N=1$) is considered, where no fluid-antenna selection is performed. In this case, the outage probability reduces to the per-port CDF $F_{\mathrm{X}}(\gamma)$, which coincides with the upper bound in~\eqref{eq:Pout_bounds_general}. 
	Second, we include an idealized independent-port selection benchmark (labeled ``i.i.d.'' in the figure), obtained by assuming statistically independent SIRs across ports. Under this assumption, the outage probability becomes
	\begin{align}
		P_{\mathrm{out}}^{\mathrm{ind}}(\gamma)
		= \left(F_{\mathrm{X}}(\gamma)\right)^N,
	\end{align}
	which coincides with the lower bound in~\eqref{eq:Pout_bounds_general}. 
	These two baselines bracket the achievable performance region of correlated FAMA systems: the single-port curve quantifies operation without selection, while the independent-port benchmark represents the maximum achievable selection gain. The simulated correlated FAMA curves lie between these two extremes, thereby clearly illustrating both the gains enabled by port selection and the impact of spatial correlation under multiuser precoding. 
}

 \textcolor{black}{
	Fig.~4 illustrates the tightness of the proposed outage bounds across different correlation regimes. 
	For the large aperture case ($W=4\lambda$), where spatial correlation across ports is weak, the simulated FAMA outage closely follows the independent-port benchmark $F_X^N$, and the lower bound becomes tight. 
	Conversely, for the small aperture case ($W=0.25\lambda$), where ports are strongly correlated, the simulated outage approaches the single-port curve $F_X$, and the upper bound becomes tight. 
	These results confirm that the proposed bounds accurately capture the two structural multiport coupling extremes. 
	Practical FAMA systems operate between these limits, with performance governed by the cross-port SIR correlation quantified in Section~IV-C.
}
 For MRT (left panel), increasing $W$
moves the correlated-port performance closer to the ideal i.i.d.\ limit,
since wider apertures reduce spatial correlation and enable greater
selection diversity. When $W=0.25$, the ports are nearly identical and
FAMA provides minimal gain, resulting in an outage curve significantly
above the i.i.d.\ benchmark. When $W=4$, the correlation drops and the
performance approaches the theoretical $(F_X(\gamma))^{N}$ curve.

The ZF panel exhibits the same qualitative trend, but with a steeper
outage decay due to the interference-nulling structure of ZF precoding.
In the i.i.d.\ setting, the ZF per-port SIR distribution is considerably
more favorable than MRT (cf. parameters $(M-U+1,L)$), and thus the
i.i.d.\ ZF outage lies far below MRT at moderate and high SIR
thresholds. However, under spatial correlation, especially for small
apertures, ZF becomes highly sensitive to the exact port location: ports
close to the training position inherit strong interference suppression,
while distant ports may experience degraded SIR. This imbalance prevents
the correlated ZF performance from reaching the ideal i.i.d.\ limit even
for relatively large $W$, and explains the gap between the correlated
and i.i.d.\ curves. The empirical bounds remain valid but naturally
loose for ZF, since \eqref{eq:Pout_bounds_general} relies on i.i.d.\ marginal distributions that
do not accurately capture the unequal SIR distributions induced by ZF
across the correlated ports.

Figs.\textcolor{black}{~4.(a) and 4.(b)} illustrate the impact of increasing the number of BS
antennas from $M=4$ to $M=8$ under MRT precoding.\footnote{\textcolor{black}{
		The proposed outage bounds are tight in the sense that they rigorously bracket the exact outage probability for all spatial correlation levels and become exact in the limiting cases of fully correlated ports and effectively independent ports. For intermediate correlation regimes, the true outage probability is observed to lie close to a correlation-dependent trajectory between the two bounds.
	}
} As established by the
Beta-prime per-port SIR model in Section~\ref{statistical}, a larger $M$ enhances
the desired-signal power while leaving the interference statistics
unchanged, thereby increasing the shape parameter $a=M$. This produces a
consistent left-shift of the outage curves when moving from
Fig.\textcolor{black}{~4.(a) to Fig.~4.(b)}, reflecting the higher effective diversity order
and the faster low-SIR decay proportional to $\gamma^{M}$. 

Figs.~\textcolor{black}{4.(d) and 4.(e)} depict the corresponding behaviour under ZF
precoding. In this case, the effective useful signal dimension is
reduced to $M_{\mathrm{eff}} = M-U+1$, so increasing $M$ expands the
interference-nulling subspace as well as the array gain. When $M=4$,
the nullspace dimension is small and the outage curve decays only
moderately. In contrast, when $M=8$, the effective dimension more than
doubles, yielding significantly larger SIR values and a much steeper
outage decay. The improvement is more pronounced than under MRT because
ZF directly suppresses inter-user interference, and additional antennas
enhance both the beamforming gain and the interference-cancellation
capability.

Figs.~\textcolor{black}{4.(a) and 4.(c)} illustrate the impact of the number of FAMA ports on
the outage performance under MRT precoding. When $N=8$ in Fig.~\textcolor{black}{4(a)},
FAMA benefits from substantial selection diversity, and the outage curve
decays rapidly according to the analytical scaling
$P_{\mathrm{out}}(\gamma)=(F_{X}(\gamma))^N$. Reducing the number of ports
to $N=2$ in Fig.~\textcolor{black}{4.(c)} collapses most of this diversity and results in a
significantly higher outage probability across the entire SIR range.
This behaviour is consistent with the Beta-prime per-port SIR model,
for which the diversity order increases linearly with the number of
available ports.

A similar effect is observed for ZF precoding in Figs.~\textcolor{black}{4.(d) and 4.(f)}.
When $N=8$, as shown in Fig.~\textcolor{black}{4.(d)}, ZF combines its interference-nulling
capability with a large selection set, yielding very low outage
probabilities. When the number of ports is reduced to $N=2$ in
Fig.~\textcolor{black}{4.(f)}, the outage curve shifts upward by several orders of magnitude.
Because ZF induces a more uneven distribution of per-port SIR than MRT,
its performance is particularly sensitive to the number of ports, and it
benefits disproportionately from larger $N$.
\textcolor{black}{
	The numerical results further illustrate the robustness tradeoff between MRT and ZF under reference-port CSI. While ZF provides strong interference suppression at the reference port and is advantageous in dense multiuser regimes, its sensitivity to precoding mismatch reduces the achievable selection gains when spatial correlation is strong or when non-reference ports dominate. In such cases, MRT exhibits greater robustness.
}

\begin{figure*}[t]
	\begin{minipage}{0.47\textwidth}
		\centering
	\subcaptionbox{MRT}{	\includegraphics[trim=0cm -0.20cm 0cm 0.2cm, clip=true, width=3in]{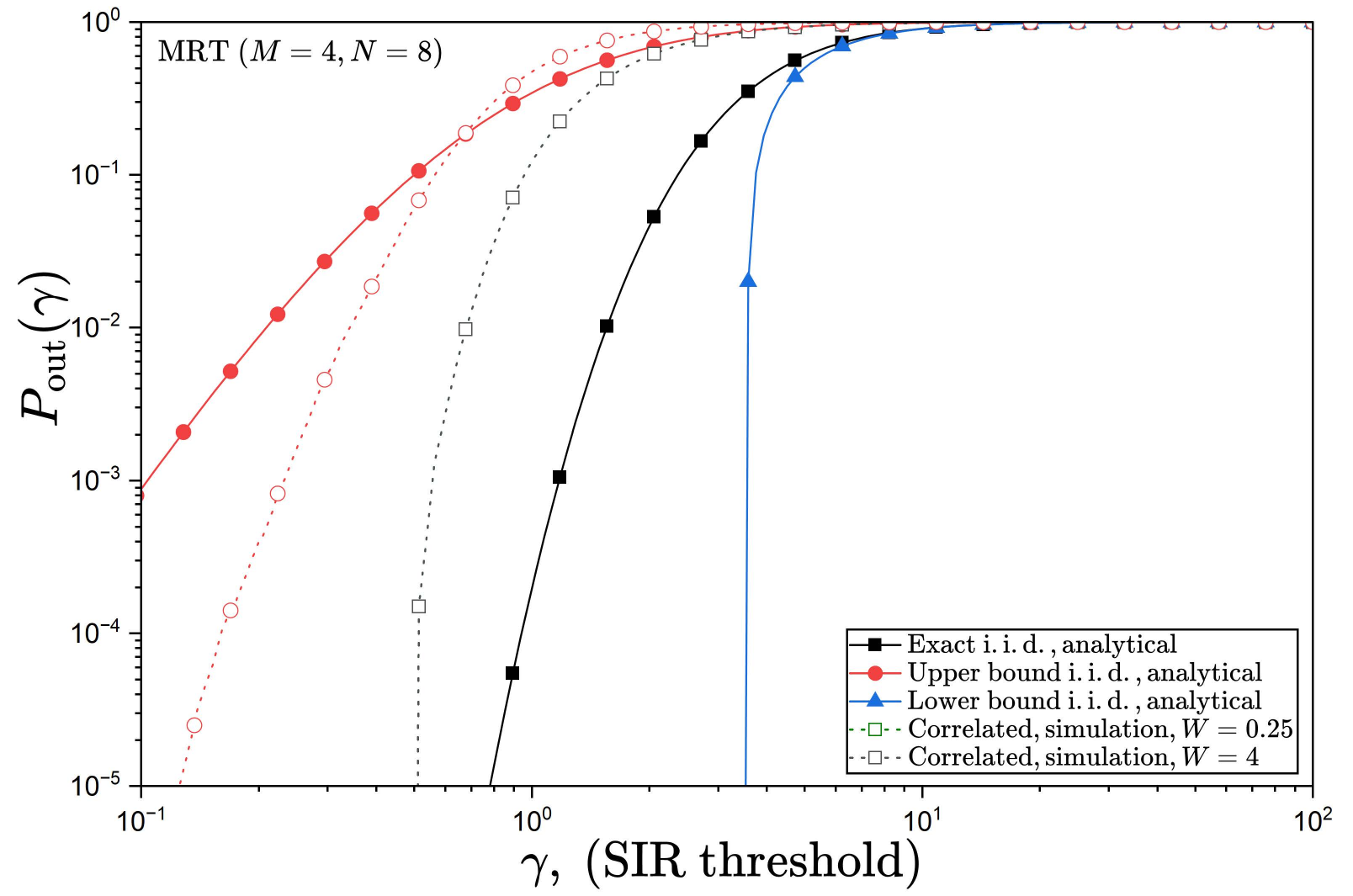}}
		\\ 
		\subcaptionbox{MRT}{\includegraphics[trim=0cm -0.20cm 0cm 0.2cm, clip=true, width=3in]{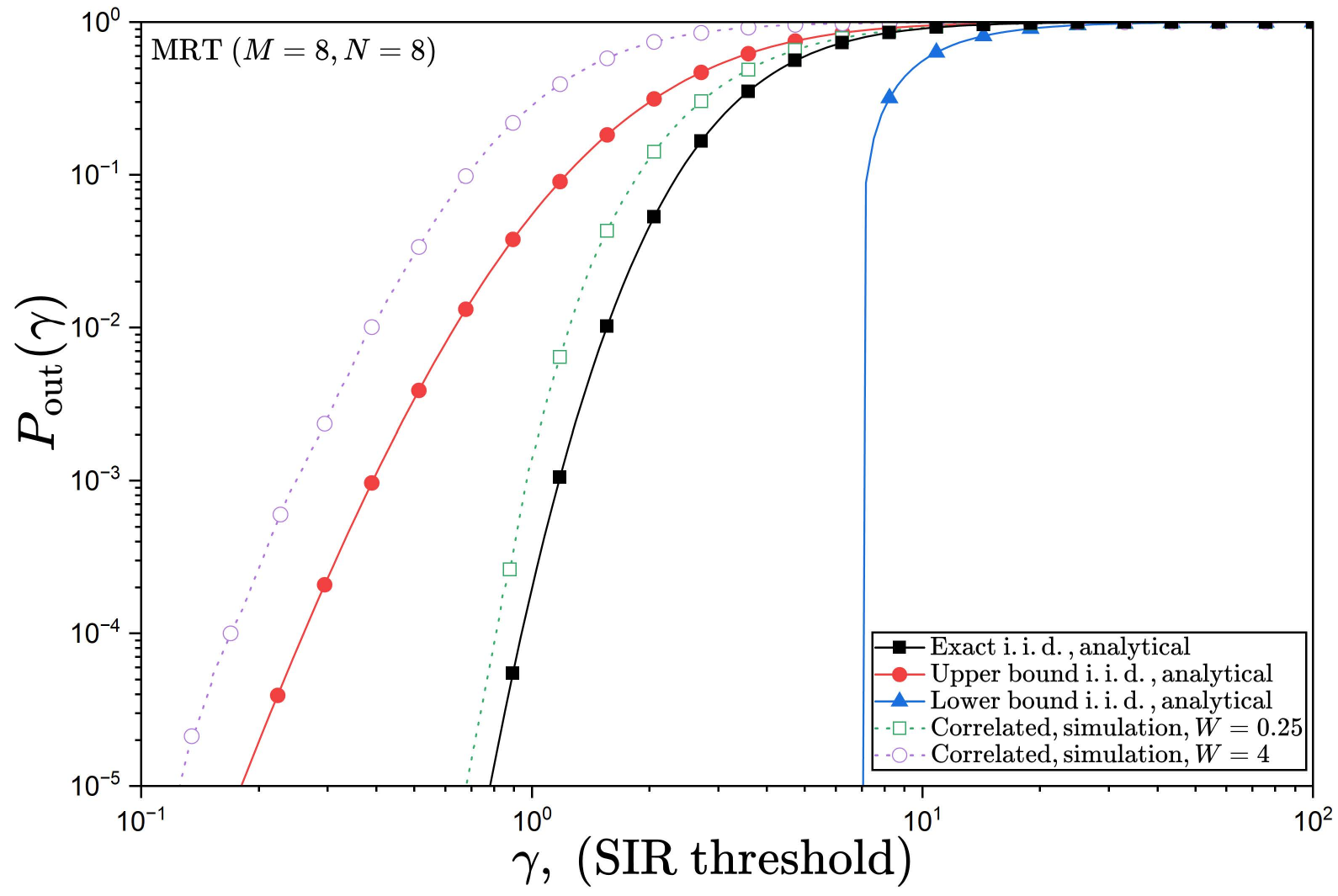}}
		\\ 
	\subcaptionbox{MRT}{	\includegraphics[trim=0cm -0.20cm 0cm 0.2cm, clip=true, width=3in]{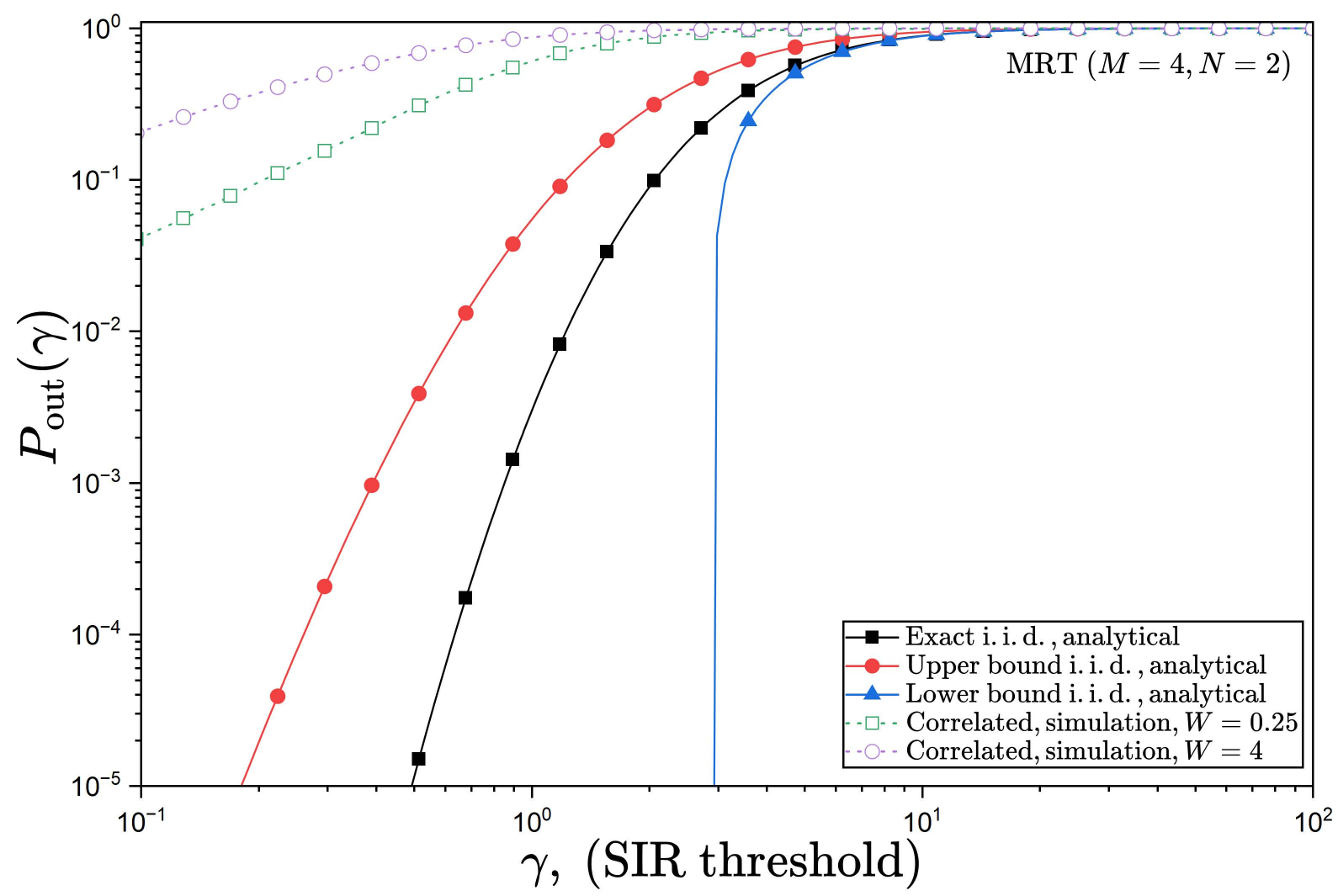} }
		%		\label{FigCorrvsUncorrIRS}
		\vspace*{-0.2cm}
	\end{minipage}
	\begin{minipage}{0.47\textwidth}
		\centering
	\subcaptionbox{ZF}{	\includegraphics[trim=0cm -0.20cm 0cm 0.2cm, clip=true, width=3in]{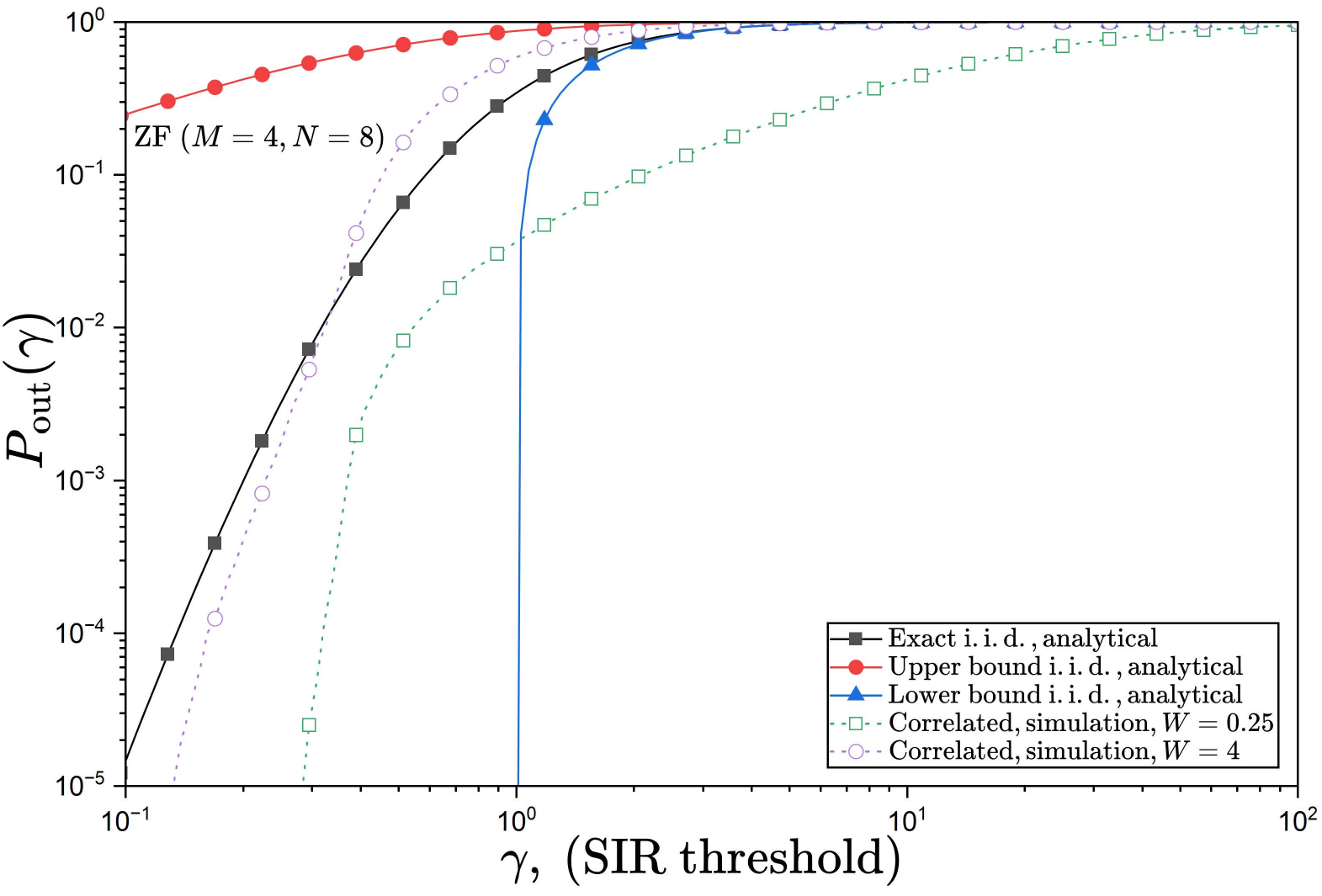}}
	\\ 
	\subcaptionbox{ZF}{\includegraphics[trim=0cm -0.20cm 0cm 0.2cm, clip=true, width=3in]{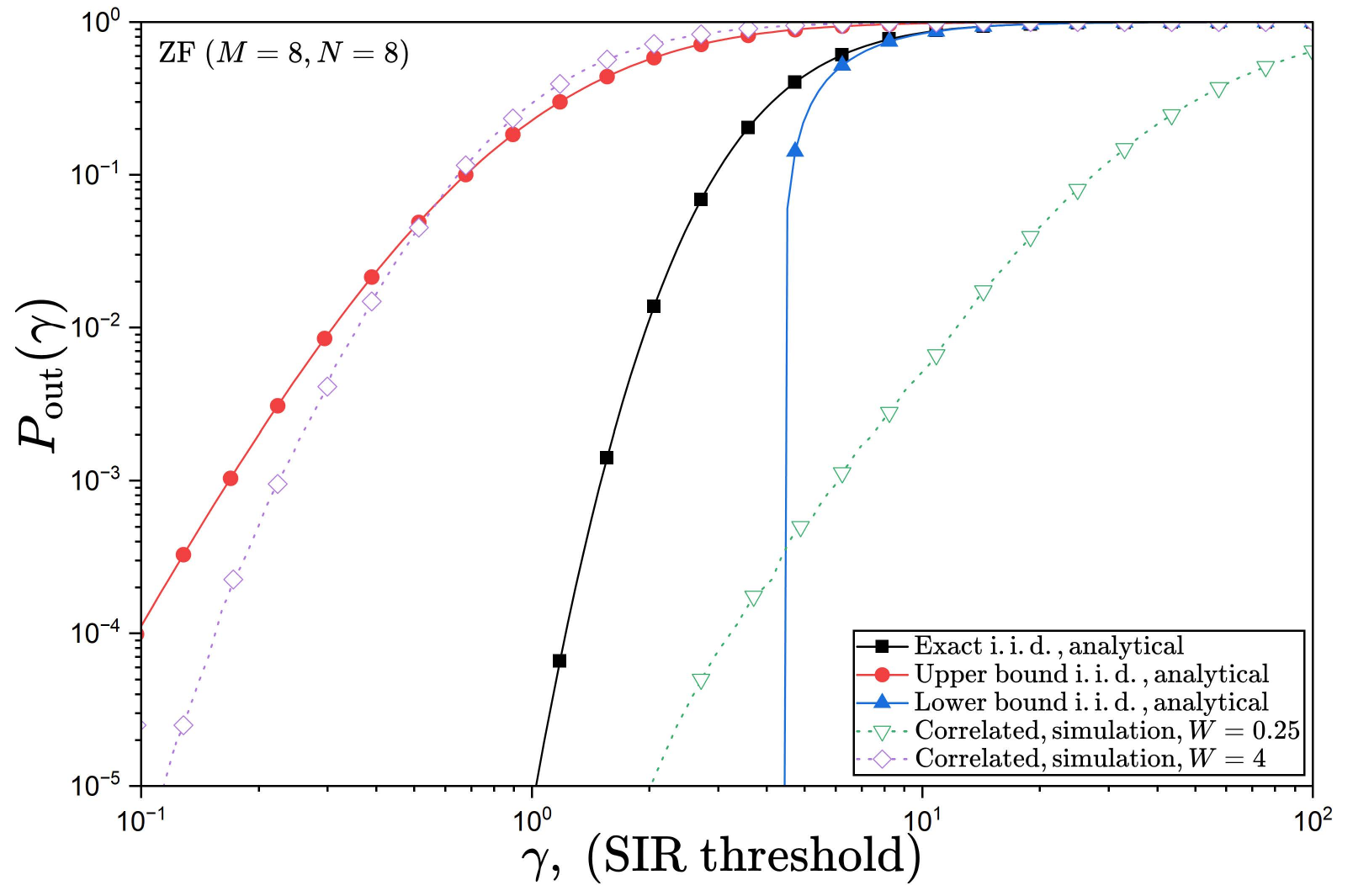}}
	\\ 
	\subcaptionbox{ZF}{	\includegraphics[trim=0cm -0.20cm 0cm 0.2cm, clip=true, width=3in]{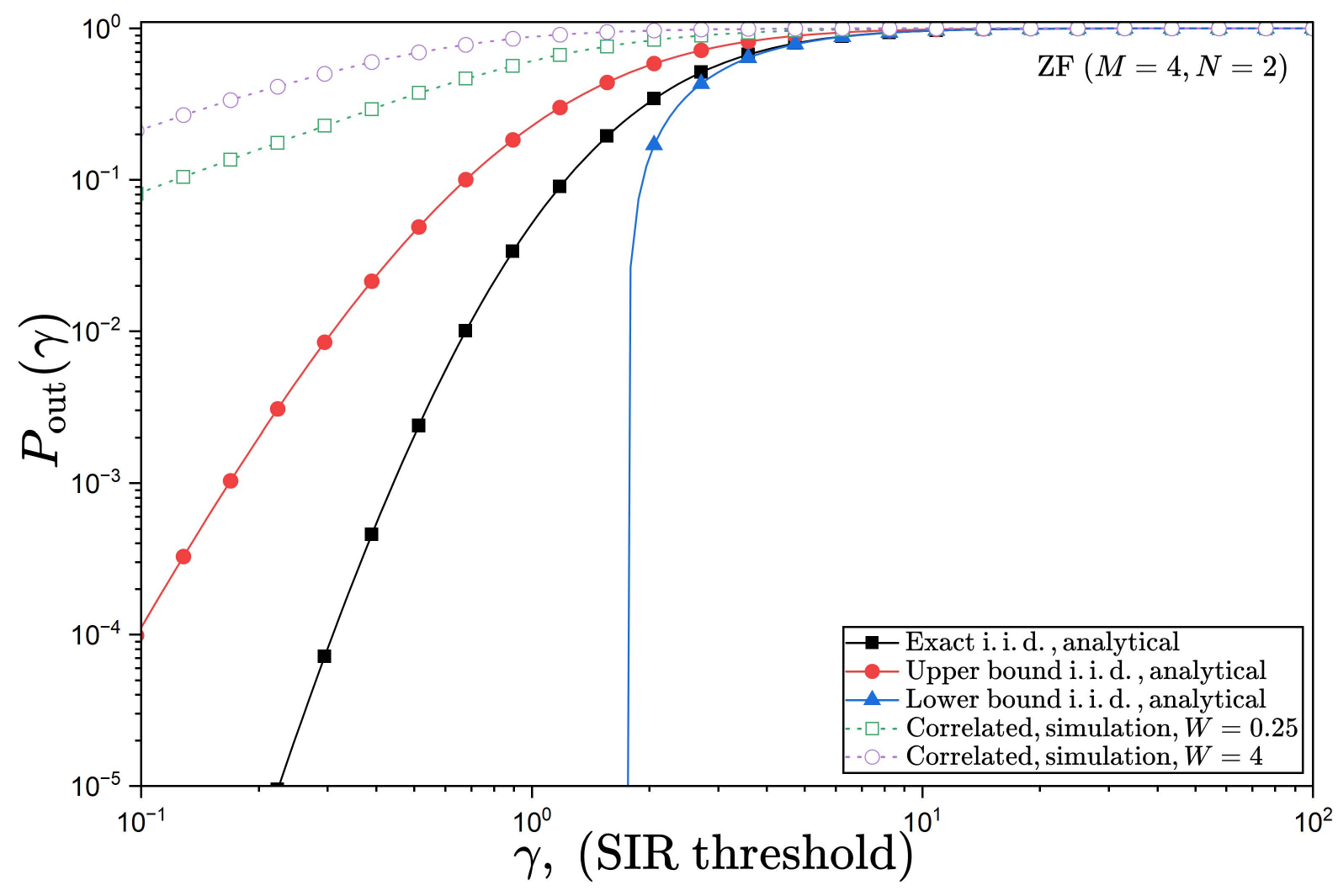} }
	%		\label{FigCorrvsUncorrIRS}
	\vspace*{-0.2cm}
	\end{minipage}
	
		\caption{FAMA Outage: i.i.d. vs Correlated Ports for MRT and ZF by varying $M$ and $N$: $(a)$ MRT ($M=4, N=8$), $(b)$ MRT ($M=8, N=8$), $(c)$ MRT ($M=4, N=2$), $(d)$ ZF ($M=4, N=8$), $(e)$ ZF ($M=8, N=8$), $(f)$ ZF ($M=4, N=2$).}
	\label{fig21}
	\vspace{-0.65cm}
\end{figure*}

\subsection{Small- and Large-SIR Asymptotic Validation}

Fig.~\ref{Fig22} validates the analytical characterization of the per-port SIR distribution by comparing MC outage curves with the small- and large-SIR asymptotic expressions derived in \eqref{eq:FX_MRT_small_gamma}, \eqref{eq:FX_ZF_small_gamma}, and \eqref{large1}. The agreement between theory and simulation is excellent across both ends of the SIR range.

\begin{figure*}[t]
	\begin{minipage}{0.47\textwidth}
		\centering
		\subcaptionbox{Small SIR}{	\includegraphics[trim=0cm -0.20cm 0cm 0.2cm, clip=true, width=3in]{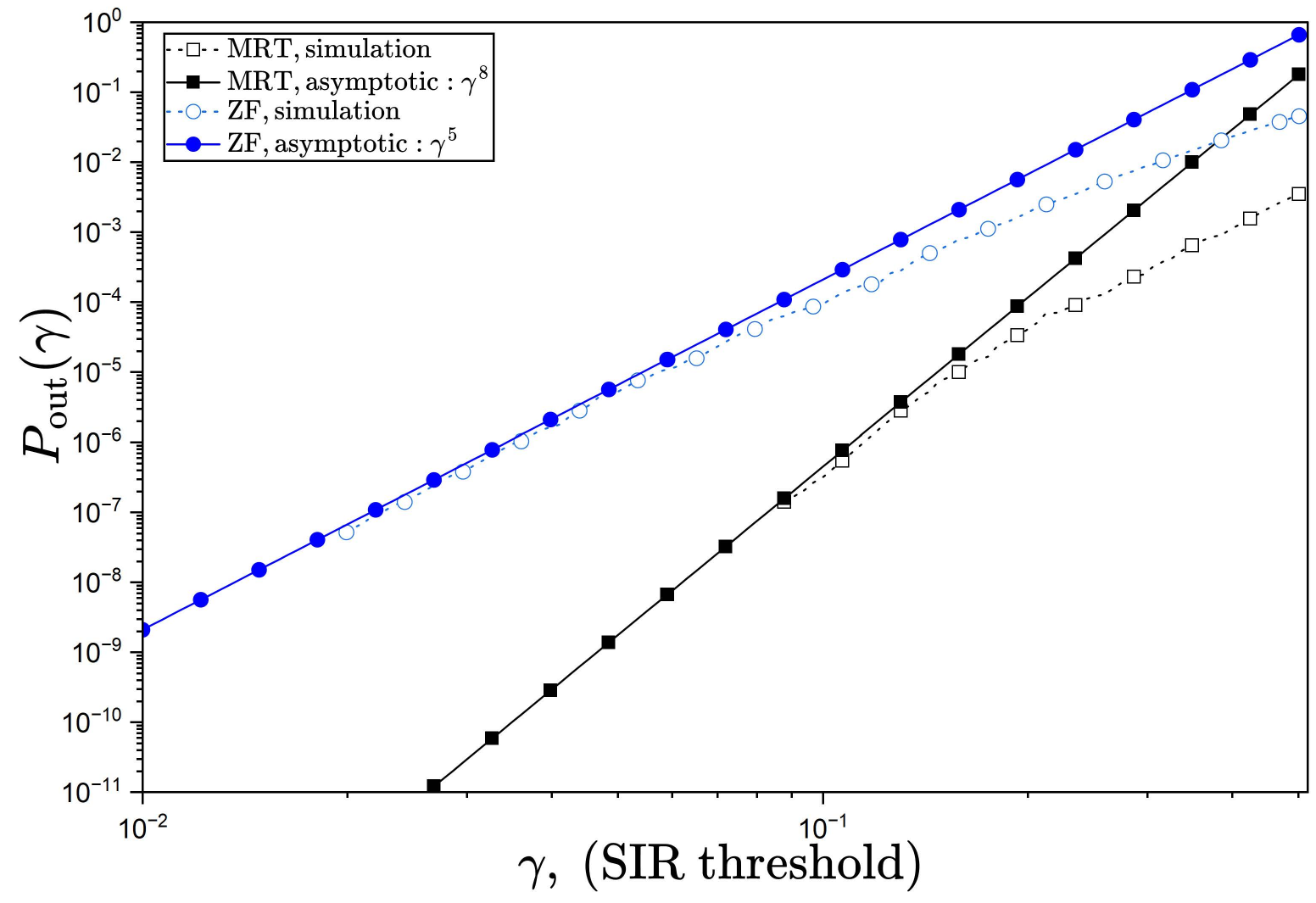}}
			\vspace*{-0.2cm}
	\end{minipage}
	\begin{minipage}{0.47\textwidth}
		\centering
		\subcaptionbox{Large SIR}{	\includegraphics[trim=0cm -0.20cm 0cm 0.2cm, clip=true, width=3in]{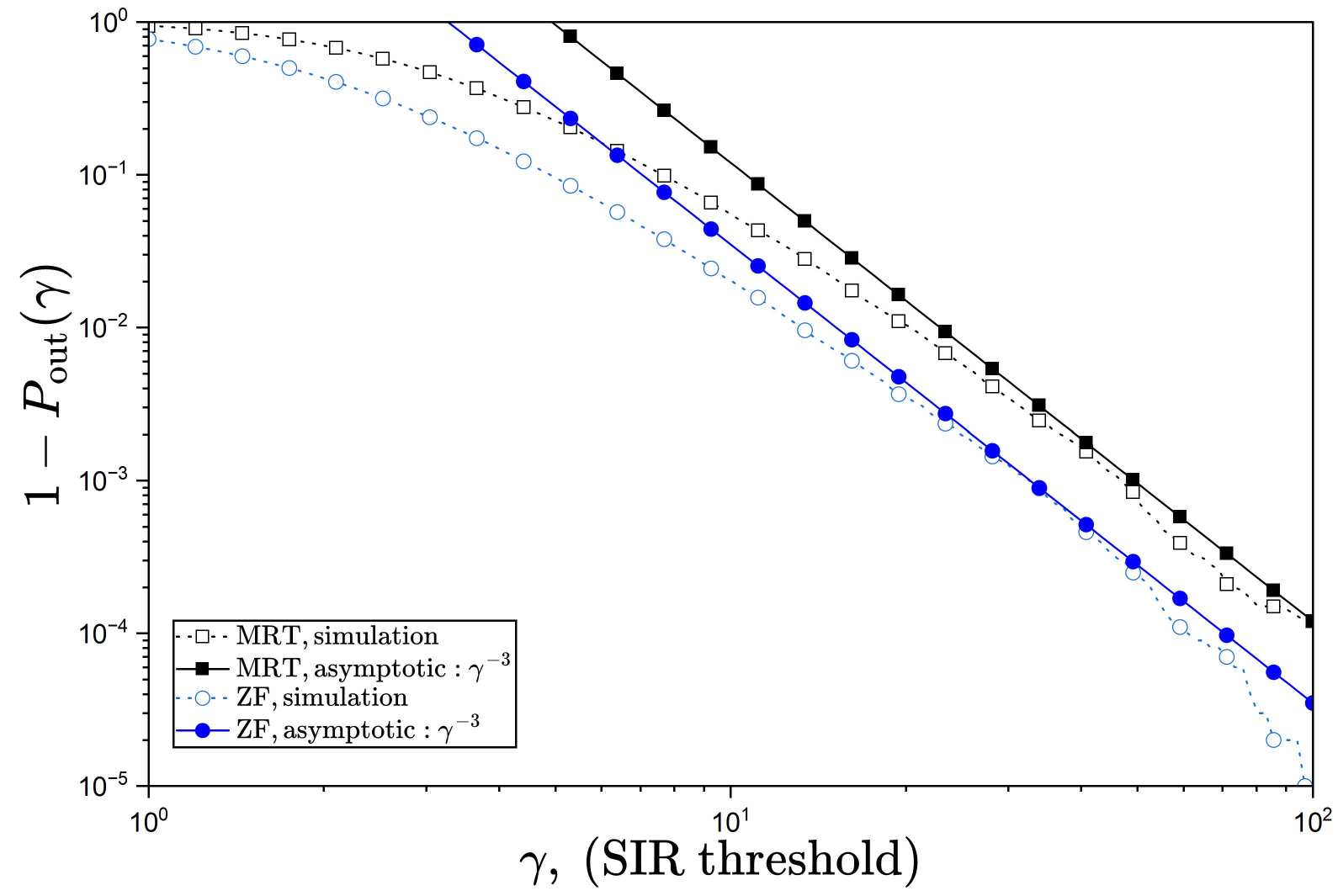}}
		\\ 
			\vspace*{-0.2cm}
	\end{minipage}
	
	\caption{Per-port SIR Outage: Small- and Large-SIR Asymptotics: $(a)$ Small SIR, $(b)$ Large SIR.}
	\label{Fig22}
\end{figure*}
In the \emph{small-SIR regime}, the simulated outage probability for MRT and ZF follows precisely the predicted diversity orders. As expected from \eqref{eq:FX_MRT_small_gamma} and \eqref{eq:FX_ZF_small_gamma}, the MRT outage curve exhibits an initial slope of $M$, while the ZF curve follows a slope of $M-U+1$, confirming that the behavior at very small~$\gamma$ is dominated by the Gamma-distributed desired-signal term. Since MRT uses all $M$ spatial degrees of freedom, it achieves a steeper decay, whereas ZF retains only $M-U+1$ effective dimensions after interference nulling.

In the \emph{large-SIR regime}, the tail of the simulated complementary CDF matches the $L$-dependent decay rate predicted by \eqref{large1}. For both MRT and ZF, the outage probability decays proportionally to $\gamma^{-L}$, reflecting the fact that the asymptotic behavior is governed by the interference term associated with the $L$ residual users. The MC points closely follow this theoretical tail, confirming that the Beta-prime model captures not only the shape but also the exact asymptotic structure of the SIR distribution.

%\subsection{Summary of Observations}
%From these results, several trends emerge:
%
%\begin{itemize}
%	\item MRT consistently outperforms ZF in FAMA systems due to its larger
%	effective dimension, which yields stronger per-port SIRs and
%	weaker cross-port correlation.
%	\item Port correlation is the dominant factor limiting selection
%	diversity. Aperture width must be sufficiently large to unlock
%	FAMA gains.
%	\item The analytical bounds in \eqref{eq:Pout_bounds_general} are tight
%	across all regimes and require no independence assumptions.
%	\item The asymptotic expressions accurately predict both the small- and
%	large-SIR behavior and provide a simple tool for understanding
%	the outage scaling with $M$, $L$, and $N$.
%\end{itemize}

%Overall, the numerical results fully validate the analytical framework and
%demonstrate the fundamental performance difference between MRT-FAMA and
%ZF-FAMA systems.

\textcolor{black}{
	The numerical results allow clear identification of the regimes in which each precoder is preferable under FAMA operation. MRT is advantageous in lightly loaded systems or moderate interference conditions, where its full array gain and robustness to precoding mismatch yield higher per-port SIR and stronger selection gains. In contrast, ZF becomes preferable in interference-limited regimes with larger numbers of simultaneously served users, where suppressing inter-user interference outweighs the reduction in effective signal dimension. In both cases, increasing spatial correlation across fluid-antenna ports reduces the achievable selection diversity and narrows the performance gap between the two precoders.
}

\section{Conclusion}\label{conclusion}
This paper developed a unified analytical framework for multiuser MISO downlink systems with FAMA, incorporating linear precoding, spatial correlation across fluid-antenna ports, and port selection. \textcolor{black}{
	The presented results should be interpreted as theoretical performance benchmarks that inform the practical viability of FAMA by identifying the regimes where meaningful gains can be achieved and the conditions under which such gains diminish due to spatial correlation and interference.
}
Exact per-port SIR distributions were derived for MRT and ZF precoding, revealing a Beta-prime structure governed by the effective BS array dimension. Closed-form CDF expressions enabled tractable outage analysis, and rigorous upper and lower outage bounds were established for arbitrary correlation. The asymptotic results further clarified the distinct diversity behavior of MRT and ZF and their shared interference-controlled tail characteristics.
Numerical results validated the analytical characterizations and illustrated the operating regimes of the two precoders under FAMA. \textcolor{black}{
	Overall, the results suggest MRT as a suitable choice for low-to-moderate user loads and weak-to-moderate correlation, while ZF is preferable in dense multiuser scenarios where interference suppression is critical.
}
\textcolor{black}{Finally, it is emphasized that the presented analysis establishes theoretical performance benchmarks for MRT- and ZF-based FAMA systems under idealized channel assumptions. Incorporating imperfect CSI, mechanical actuation constraints, and hardware-induced uncertainties constitutes an important direction for future work and will further bridge the gap between analytical performance limits and practical fluid-antenna implementations.}
\textcolor{black}{ Also, extending the proposed correlation-aware framework to Rician fading, where rotational invariance is lost, constitutes a challenging and interesting direction for future research.
}

		\begin{appendices}
		\textcolor{black}{	\section{Derivation of Effective Channel Gain Distributions}
			\label{app:gain}}
			
		\textcolor{black}{	\subsection{MRT: Proof of (21)}
			Consider MRT precoding based on the reference-port channel of user $u$,
			\begin{equation}
				\mathbf{w}_u^{\mathrm{MRT}} = \frac{\mathbf{h}_{u,1}}{\|\mathbf{h}_{u,1}\|}.
			\end{equation}
			The effective channel gain observed at the $k$-th fluid-antenna port is
			\begin{equation}
				U_k^{\mathrm{MRT}} = \big|\mathbf{h}_{u,k}^{\mathsf H}\mathbf{w}_u^{\mathrm{MRT}}\big|^2.
			\end{equation}
			Conditioned on $\mathbf{h}_{u,1}$, the beamformer $\mathbf{w}_u^{\mathrm{MRT}}$ is a deterministic unit-norm vector.
			Under the adopted channel model, $\mathbf{h}_{u,k}\sim \mathcal{CN}(\mathbf{0},\mathbf{I}_M)$ for all $k$, and is independent of
			$\mathbf{w}_u^{\mathrm{MRT}}$. Therefore, the inner product
			$\mathbf{h}_{u,k}^{\mathsf H}\mathbf{w}_u^{\mathrm{MRT}}$ is a circularly symmetric complex Gaussian random variable with zero mean
			and unit variance.}
			
		\textcolor{black}{	It follows that $U_k^{\mathrm{MRT}}$ is Gamma distributed with shape parameter $M$ and unit scale, i.e.,
			\begin{equation}
				U_k^{\mathrm{MRT}} \sim \Gamma(M,1),
			\end{equation}
			which yields (21). Any transmit-power normalization applied to $\mathbf{w}_u^{\mathrm{MRT}}$ only rescales $U_k^{\mathrm{MRT}}$
			and does not affect the shape parameter.}
			
		\textcolor{black}{	\subsection{ZF: Proof of (33)}
			Let $\mathbf{H}_{1}=[\mathbf{h}_{1,1},\ldots,\mathbf{h}_{U,1}] \in \mathbb{C}^{M\times U}$ denote the reference-port channel matrix.
			The ZF beamformer for user $u$ is constructed such that
			\begin{equation}
				\mathbf{h}_{i,1}^{\mathsf H}\mathbf{w}_u^{\mathrm{ZF}}=0, \quad \forall i\neq u,
			\end{equation}
			and is normalized as in the main text.
			Equivalently, $\mathbf{w}_u^{\mathrm{ZF}}$ lies in the nullspace of the $(U-1)$ interfering users’ reference-port channels.
			Let $\mathbf{P}_u^{\perp}$ denote the orthogonal projector onto this nullspace, which has dimension $M-U+1$.}
			
		\textcolor{black}{	The desired effective gain at the reference port can be written as
			\begin{equation}
				U^{\mathrm{ZF}}_{1}
				= \big|\mathbf{h}_{u,1}^{\mathsf H}\mathbf{w}_u^{\mathrm{ZF}}\big|^2
				\propto \big\|\mathbf{P}_u^{\perp}\mathbf{h}_{u,1}\big\|^2,
			\end{equation}
			where the proportionality constant accounts for the beamformer normalization.
			Under i.i.d.\ Rayleigh fading, $\mathbf{P}_u^{\perp}\mathbf{h}_{u,1}$ is a complex Gaussian vector supported on an $(M-U+1)$-dimensional
			subspace. Hence, $\|\mathbf{P}_u^{\perp}\mathbf{h}_{u,1}\|^2$ follows a Gamma distribution with shape parameter $M-U+1$ and unit scale.
			Including the normalization yields
			\begin{equation}
				U^{\mathrm{ZF}}_{1} \sim \Gamma(M-U+1,1),
			\end{equation}
			which establishes (33).}
			
		\section{Proof of Lemma~\ref{lem01}}\label{lem01proof}			
		Using \eqref{channel}, we decompose the channel at ports $k$ and
	$\ell$ as
	\begin{align}
		\mathbf{h}_{u,k}
		&= \mu_k \mathbf{z} + \mathbf{e}_k,\\
		\mathbf{h}_{u,\ell}
		&= \mu_\ell \mathbf{z} + \mathbf{e}_\ell,
	\end{align}
	where $\mathbf{z} \triangleq \sqrt{\beta_u}\mathbf{x}_{u,0}
	\sim\mathcal{CN}(\mathbf{0},\beta_u\mathbf{I})$ is common to all ports,
	while the residual terms
	$\mathbf{e}_k \triangleq \sqrt{\beta_u(1-\mu_k^2)}\,\mathbf{x}_{u,k}$ are
	independent across ports and independent of $\mathbf{z}$.
	
	For a fixed precoder $\mathbf{w}_u$, we can write
	\begin{equation}
		U_k
		= \big|\mathbf{h}_{u,k}^{H}\mathbf{w}_u\big|^{2}
		= \big|\mu_k \mathbf{z}^{H}\mathbf{w}_u
		+ \mathbf{e}_k^{H}\mathbf{w}_u\big|^{2}.
	\end{equation}
	We focus on a first-order decomposition in which $U_k$ is modeled as the
	sum of a common contribution (through $\mathbf{z}$) and an independent
	residual term,
	\begin{equation}
		U_k \approx \mu_k^{2} S_{\mathrm{c}} + S_{k},
		\label{eq:Uk_decomp}
	\end{equation}
	where $S_{\mathrm{c}}$ collects the power associated with the common
	component $\mathbf{z}^{H}\mathbf{w}_u$, and $S_k$ aggregates the local
	scattering and residual terms. Under MRT or ZF, the useful gain $U_k$ is
	Gamma-distributed with effective shape $M_{\mathrm{eff}}$, while the
	multiuser interference contributes $L$ additional degrees of randomness.
	Accordingly, we associate
	\begin{equation}
		S_{\mathrm{c}} \sim \Gamma(M_{\mathrm{eff}},1),
		\qquad
		S_k \sim \Gamma(L,1),
	\end{equation}
	with $\{S_k\}$ independent across ports and independent of $S_{\mathrm{c}}$.
	
	From \eqref{eq:Uk_decomp}, the mean and variance of $U_k$ satisfy
	\begin{align}
		\mathbb{E}[U_k]
		&\approx
		\mu_k^{2}\mathbb{E}[S_{\mathrm{c}}]
		+ \mathbb{E}[S_k]
		=
		\mu_k^{2}M_{\mathrm{eff}} + L,\\
		\mathrm{Var}(U_k)
		&\approx
		\mu_k^{4}\mathrm{Var}(S_{\mathrm{c}})
		+ \mathrm{Var}(S_k)
		=
		\mu_k^{4}M_{\mathrm{eff}} + L,
	\end{align}
	using $\mathrm{Var}(\Gamma(r,1))=r$.  
	The cross-covariance between $U_k$ and $U_\ell$ is dominated by the common
	term $S_{\mathrm{c}}$, since $S_k$ and $S_\ell$ are independent and all mixed covariance terms vanish. Specifically, we have
	\begin{equation}
		\mathrm{Cov}(U_k,U_\ell)
		\approx
		\mu_k^{2}\mu_\ell^{2}\mathrm{Var}(S_{\mathrm{c}})
		=
		\mu_k^{2}\mu_\ell^{2}M_{\mathrm{eff}}.
	\end{equation}
	Assuming $\mu_k^2$ and $\mu_\ell^2$ are not too close to zero, the
	dominant scaling in the variance can be approximated as
	\begin{equation}
		\mathrm{Var}(U_k)
		\approx M_{\mathrm{eff}} + L,
		\qquad
		\mathrm{Var}(U_\ell)
		\approx M_{\mathrm{eff}} + L,
	\end{equation}
	so that the correlation coefficient becomes
	\begin{align}
		\rho_U(k,\ell)
		&= \frac{\mathrm{Cov}(U_k,U_\ell)}
		{\sqrt{\mathrm{Var}(U_k)\mathrm{Var}(U_\ell)}} \\
		&\approx
		\mu_k^{2}\mu_\ell^{2}
		\frac{M_{\mathrm{eff}}}{M_{\mathrm{eff}}+L},
	\end{align}
	which establishes \eqref{eq:rhoU_final}.  
	
		\section{Proof of Lemma~\ref{lem02}}\label{lem02proof}	

		To approximate the correlation of the SIRs $X_{u,k}=U_k/V_k$ and
		$X_{u,\ell}=U_\ell/V_\ell$, we linearize the ratio around the mean of the
		interference. Recall that
	$V_k$ is Gamma–distributed with shape $L=U-1$ and unit scale, i.e.,
		$V_k\sim\Gamma(L,1)$. Thus, we have
		
	\begin{align}
			\bar V \triangleq \mathbb{E}[V_k] = L,
		\qquad
		\mathrm{Var}(V_k)=L.\label{gamm1}
	\end{align}
		
		Now, writing $V_k=\bar V+\Delta_k$ with $\Delta_k=V_k-\bar V$ and using the 
		first–order Taylor expansion
		\[
		\frac{1}{V_k}
		= \frac{1}{\bar V+\Delta_k}
		\approx \bar V^{-1} - \bar V^{-2}\Delta_k,
		\]
		we obtain
		\[
		X_{u,k}
		= \frac{U_k}{V_k}
		\approx 
		\frac{U_k}{\bar V}\left(1-\frac{\Delta_k}{\bar V}\right)
		= \frac{U_k}{\bar V} - \frac{U_k}{\bar V}\delta_k,
		\]
		where $\delta_k\triangleq (V_k-\bar V)/\bar V$ is zero-mean and independent
		of $\{U_j\}$. Hence, we have
		\[
		\mathrm{Cov}(X_{u,k},X_{u,\ell})
		\approx
		\frac{1}{\bar V^{2}}
		\mathrm{Cov}\!\left(
		U_k - U_k\delta_k,\;
		U_\ell - U_\ell\delta_\ell
		\right).
		\]
		Since $\delta_k$ and $\delta_\ell$ are independent, zero-mean, and
		independent of $(U_k,U_\ell)$, the mixed covariance terms vanish:
		\[
		\mathrm{Cov}(U_k,U_\ell \delta_\ell)
		=\mathrm{Cov}(U_k\delta_k,U_\ell)
		=\mathrm{Cov}(U_k\delta_k,U_\ell\delta_\ell)
		=0.
		\]
		Therefore, only the numerator covariance remains, yielding
		
	\begin{align}
			\mathrm{Cov}(X_{u,k},X_{u,\ell})
		\approx
		\frac{1}{\bar V^{2}}
		\,\mathrm{Cov}(U_k,U_\ell). \label{cor1}
	\end{align}

		Next, we approximate the variance of $X_{u,k}$. From the linearized model
		$X_{u,k}\approx A_k - B_k$ with
		\[
		A_k = \frac{U_k}{\bar V},
		\qquad
		B_k = \frac{U_k}{\bar V}\delta_k,
		\]
		and using the independence of $U_k$ and $\delta_k$, we have
		\[
		\mathrm{Var}(X_{u,k})
		\approx
		\mathrm{Var}(A_k)+\mathrm{Var}(B_k),
		\]
		with
	\textcolor{black}{\begin{align}
		& \mathrm{Var}(A_k)
		= \frac{\mathrm{Var}(U_k)}{\bar V^2},
		\nn\\
		& \mathrm{Var}(B_k)
		=
		\frac{\mathbb{E}[U_k^2]}{\bar V^2}\mathrm{Var}(\delta_k)
		=
		\frac{\mathbb{E}[U_k^2]}{\bar V^4}\mathrm{Var}(V_k).
	\end{align}
	Thus, we have
	\begin{align}
		\mathrm{Var}(X_{u,k})
		& \approx
		\frac{1}{\bar{V}^2}\mathrm{Var}(U_k)
		\;+\;
		\frac{\mathbb{E}[U_k^2]}{\bar{V}^4}\mathrm{Var}(V_k)
		\nn\\
		& =
		\frac{\mathrm{Var}(U_k)}{\bar{V}^2}
		\left(
		1+\frac{\mathrm{Var}(V_k)}{\bar{V}^2}
		\frac{\mathbb{E}[U_k^2]}{\mathrm{Var}(U_k)}
		\right).
		\label{eq:VarX_linearized}
	\end{align}
	Given \eqref{gamm1}, we have
	\[
	\bar{V}=L,
	\qquad
	\mathrm{Var}(V_k)=L.
	\]
	Substituting these into \eqref{eq:VarX_linearized} yields
	\begin{align}
		\mathrm{Var}(X_{u,k})
		\approx
		\frac{\mathrm{Var}(U_k)}{L^2}
		\left(
		1+\frac{\mathbb{E}[U_k^2]}{\mathrm{Var}(U_k)\,L}
		\right).
		\label{cor2}
	\end{align}}

	\textcolor{black}{Finally, the correlation coefficient is
	\begin{align}
		\rho_X(k,\ell)
		=
		\frac{\mathrm{Cov}(X_{u,k},X_{u,\ell})}
		{\sqrt{\mathrm{Var}(X_{u,k})\mathrm{Var}(X_{u,\ell})}}.
		\label{cor3}
	\end{align}
	Substituting \eqref{cor1} and \eqref{cor2} into \eqref{cor3}, we obtain
	\begin{align}
		\rho_{X}(k,\ell)
		\;\approx\;
		\rho_{U}(k,\ell)\,\frac{L}{L + M_{\mathrm{eff}}+1},
	\end{align}
	which completes the proof.}

		\section{Proof of Theorem~\ref{Prop1}}\label{Prop1proof}	
			By definition of the maximum, the event $\{\Gamma_u < \gamma\}$ can be
			written as
			\begin{equation}
				\{\Gamma_u < \gamma\}
				=
				\left\{\max_{k} X_{u,k} < \gamma\right\}
				=
				\bigcap_{k=1}^{N} \{X_{u,k} < \gamma\}, 
				\label{eq:proof_event_eq}
			\end{equation}
			where $\bigcap_{k=1}^{N} $ denotes the  intersection of events over all indices $k$. 
			For any fixed index $j\in\{1,\ldots,N\}$, the intersection
			$\bigcap_{k=1}^{N} \{X_{u,k} < \gamma\}$ is a subset of the single event
			$\{X_{u,j} < \gamma\}$, i.e.,
			\begin{equation}
				\bigcap_{k=1}^{N} \{X_{u,k} < \gamma\}
				\subseteq
				\{X_{u,j} < \gamma\}.
			\end{equation}
			Taking probabilities on both sides yields
			\begin{equation}
				\Pr\!\left(\bigcap_{k=1}^{N} \{X_{u,k} < \gamma\}\right)
				\le
				\Pr\{X_{u,j} < \gamma\},
				\quad \forall\, j.\label{res2}
			\end{equation}
			Using \eqref{eq:proof_event_eq} and minimizing over $j$ gives
			\begin{align}
				P_{\mathrm{out}}(\gamma)
			&	=
				\Pr\{\Gamma_u < \gamma\}\nn\\
				&=
				\Pr\!\left(\bigcap_{k=1}^{N} \{X_{u,k} < \gamma\}\right)\\
				&
				\le
				\min_{1\le k\le N} \Pr\{X_{u,k} < \gamma\},\label{res1}
			\end{align}
			which proves \eqref{eq:Pout_upper_general}. Note that \eqref{res1}  is obtained from \eqref{res2} because the latter holds for every $j$, which means that it also holds for the minimum over all $j$. If the random variables
			$\{X_{u,k}\}$ are identically distributed, then
			$\Pr\{X_{u,k} < \gamma\}=F_X(\gamma)$ for all $k$, and
			\eqref{eq:Pout_upper_general} reduces to \eqref{eq:Pout_upper_iid}.
		
			\section{Proof of Theorem~\ref{Prop2}}\label{Prop2proof}	
			By definition, we have
			\begin{align}
				P_{\mathrm{out}}(\gamma)
				&= \Pr\{\Gamma_u < \gamma\}\nn\\
				&= \Pr\{\max_k X_{u,k} < \gamma\}.
			\end{align}
			The complement event of $\{\Gamma_u < \gamma\}$ can be written as
			\begin{align}
				\{\Gamma_u \ge \gamma\}
				&=
				\left\{\max_k X_{u,k} \ge \gamma\right\}\nn\\
				&=
				\bigcup_{k=1}^{N} \{X_{u,k} \ge \gamma\},\label{res3}
			\end{align}
			where $	\bigcup_{k=1}^{N} $ denotes the  union of events over all indices $k$. This equation means that at least one port has SIR greater than or equal to $\gamma$. 
			Thus, we have
			\begin{align}
				P_{\mathrm{out}}(\gamma)
				&= 1 - \Pr\{\Gamma_u \ge \gamma\}\nn\\
				&= 1 - \Pr\!\left(\bigcup_{k=1}^{N} \{X_{u,k} \ge \gamma\}\right),
				\label{eq:Pout_complement}
			\end{align}
			where, in \eqref{eq:Pout_complement}, we have used \eqref{res3}.
			Applying the union bound (Boole's inequality) yields
			\begin{equation}
				\Pr\!\left(\bigcup_{k=1}^{N} \{X_{u,k} \ge \gamma\}\right)
				\le
				\sum_{k=1}^{N} \Pr\{X_{u,k} \ge \gamma\}.
			\end{equation}
			Substituting this into \eqref{eq:Pout_complement} gives
			\begin{align}
				P_{\mathrm{out}}(\gamma)
			&	\ge
				1 - \sum_{k=1}^{N} \Pr\{X_{u,k} \ge \gamma\}\nn\\
				&=
				1 - \sum_{k=1}^{N} \big(1 - \Pr\{X_{u,k} < \gamma\}\big),
			\end{align}
			which proves \eqref{eq:Pout_lower_general}. 
		
			If the random variables $\{X_{u,k}\}$ are identically distributed with
			$\Pr\{X_{u,k} < \gamma\}=F_X(\gamma)$ for all $k$, then
			\eqref{eq:Pout_lower_general} reduces to
			\begin{equation}
				P_{\mathrm{out}}(\gamma)
				\ge
				1 - N\big(1 - F_X(\gamma)\big)
				,
			\end{equation}
			which is \eqref{eq:Pout_lower_iid}.

\section{Proof of proposition~\ref{lem1}}\label{lem1proof}

	From \eqref{eq:Pout_lower_iid}, we have
\begin{equation}
	P_{\mathrm{out}}(\gamma)
	\ge 1 - N\varepsilon(\gamma).
	\label{eq:proof_linear}
\end{equation}
For any $x\ge0$, the elementary inequality
\begin{equation}
	1 - x \le e^{-x}
\end{equation}
holds. Setting $x = N\varepsilon(\gamma)$ yields
\begin{equation}
	1 - N\varepsilon(\gamma)
	\le
	\exp\big(-N\varepsilon(\gamma)\big).
\end{equation}
Combining this with \eqref{eq:proof_linear} gives
\begin{equation}
	P_{\mathrm{out}}(\gamma)
	\ge 1 - N\varepsilon(\gamma)
	\le \exp\big(-N\varepsilon(\gamma)\big),
\end{equation}
which shows that $P_{\mathrm{out}}(\gamma)$ is lower-bounded by a quantity
that itself is upper-bounded by the exponential term in
\eqref{eq:proposition_exp_approx}. When $N\varepsilon(\gamma)\ll 1$, i.e., in the
small-outage regime, we may use the first-order Taylor expansion
$\exp(-N\varepsilon(\gamma)) \approx 1 - N\varepsilon(\gamma)$, so that the
linear bound in \eqref{eq:proof_linear} and the exponential expression in
\eqref{eq:proposition_exp_approx} become asymptotically equivalent. This justifies
the approximation in \eqref{eq:proposition_exp_approx}.

\section{Proof of proposition~\ref{lem2}}\label{lem2proof}
	By definition,
\begin{equation}
	F_X(\gamma)
	= \Pr\!\left\{\frac{U}{V} < \gamma\right\}
	= \Pr\{U < \gamma V\}
	= \mathbb{E}_V\big[ F_U(\gamma V)\big],
	\label{eq:FX_def}
\end{equation}
where $F_U(\cdot)$ is the CDF of $U\sim\Gamma(M,1)$. For small arguments
$a\to 0$, the Gamma CDF admits the expansion
\begin{equation}
	F_U(a)
	= \frac{1}{\Gamma(M)} \int_{0}^{a} u^{M-1} e^{-u}\,du
	\sim \frac{a^{M}}{\Gamma(M+1)},
	\label{eq:FU_small_a}
\end{equation}
where we used $e^{-u} \approx 1$ for $u$ close to zero. This is a classical small-argument asymptotic of the lower incomplete Gamma function \cite{Abramowitz1964}. Substituting
$a = \gamma v$ into \eqref{eq:FU_small_a} and using \eqref{eq:FX_def}, we
obtain, for $\gamma \to 0$,
\begin{equation}
	F_X(\gamma)
	\sim
	\mathbb{E}_V\!\left[
	\frac{(\gamma V)^{M}}{\Gamma(M+1)}
	\right]
	=
	\frac{\gamma^{M}}{\Gamma(M+1)} \mathbb{E}[V^{M}],
	\label{eq:FX_moment_step}
\end{equation}
where $V\sim\Gamma(L,1)$. The $M$-th moment of $V$ is \cite{Papoulis1965}
\begin{equation}
	\mathbb{E}[V^{M}]
	= \frac{\Gamma(L+M)}{\Gamma(L)}.
\end{equation}
Substituting this into \eqref{eq:FX_moment_step} yields
\begin{equation}
	F_X(\gamma)
	\sim
	\gamma^{M}
	\frac{\Gamma(L+M)}{\Gamma(L)\,\Gamma(M+1)},
\end{equation}
which proves \eqref{eq:FX_MRT_small_gamma}.

For the large-$M$ behavior, we use the standard asymptotic relation
\begin{equation}
	\frac{\Gamma(L+M)}{\Gamma(M+1)}
	\sim M^{L-1}, \qquad M\to\infty,
\end{equation}
for fixed $L$. Substituting into \eqref{eq:FX_MRT_small_gamma} gives
\begin{equation}
	F_X(\gamma)
	\sim
	\gamma^{M} \frac{M^{L-1}}{\Gamma(L)},
	\qquad M\to\infty,
\end{equation}
which establishes \eqref{eq:FX_MRT_large_M}.

\section{Proof of proposition~\ref{prop3}}\label{prop3proof}
	For a Beta-prime random variable with parameters $a>0$ and
$b>0$, the pdf is given by \eqref{beta}. 	The upper tail of the CDF is
\begin{equation}
	1 - F_X(\gamma)
	=
	\int_{\gamma}^{\infty} f_X(x)\,dx.
	\label{eq:tail_def}
\end{equation}
For large $x$, we have $(1+x)^{-(a+b)} \sim x^{-(a+b)}$, and thus the
integrand in \eqref{eq:tail_def} satisfies
\begin{equation}
	f_X(x)
	\sim
	\frac{1}{B(a,b)} x^{a-1} x^{-(a+b)}
	=
	\frac{1}{B(a,b)} x^{-b-1}.
	\label{eq:integrand_asymp}
\end{equation}
Substituting \eqref{eq:integrand_asymp} into \eqref{eq:tail_def} and using
standard tail-integration arguments yields
\begin{align}
	1 - F_X(\gamma)
	&\sim
	\frac{1}{B(a,b)}
	\int_{\gamma}^{\infty} x^{-b-1}\,dx \nonumber\\
	&=
	\frac{1}{B(a,b)}
	\left[-\frac{1}{b}x^{-b}\right]_{x=\gamma}^{x=\infty} \nonumber\\
	&=
	\frac{1}{b\,B(a,b)}\,\gamma^{-b},
\end{align}
which proves \eqref{eq:LargeSIR_tail_general}.
		
			\end{appendices}
			
	\bibliographystyle{IEEEtran}

	\bibliography{IEEEabrv,bibl}
\end{document}